\documentclass[sigplan, fontsize=10pt]{acmart}
\renewcommand\footnotetextcopyrightpermission[1]{} 
\acmConference[]{}{}{}
\acmYear{}
\acmISBN{} 
\acmDOI{} 
\startPage{1}

\setcopyright{none}

\usepackage{booktabs}   
\usepackage{subcaption} 

\usepackage{listings, tikz, tikz-cd, enumitem, amsthm, bcprules, enumitem,
  pgfplots, }
\usepackage[linesnumbered,ruled]{algorithm2e}
\usepackage[linewidth=0.5pt]{mdframed}
\tikzset{>=latex}

\newtheorem{remark}{Remark}

\usetikzlibrary{patterns,calc,backgrounds}

\tikzstyle{nnf}=[
  >=stealth,font=\small,auto,scale=0.7,every node/.style={scale=0.7}
]
\tikzstyle{extnode}=[
  draw,circle,inner sep=2pt,fill=white
]

\tikzstyle{leafnode}=[
  draw,fill=gray!20,inner sep=3.5pt
]
\tikzstyle{constnode}=[
  draw,fill=white,inner sep=3.5pt
]
\tikzstyle{label}=[
  fill=white,inner sep=2.5pt
]

\tikzstyle{acarrow}=[
    decoration={markings,mark=at position 1 with {\arrow[scale=0.6]{>}}},
    postaction={decorate},
    shorten >=0.4pt,
    >=latex,
    line width=0.1
]

\tikzstyle{bnarrow}=[
    decoration={markings,mark=at position 1 with {\arrow[scale=1.5]{>}}},
    postaction={decorate},
    shorten >=0.7pt,
    >=latex,
    line width=0.3
]
\tikzstyle{bayesnet}=[
  >=latex, thick, auto
]
\tikzstyle{bnnode}=[
  draw,ellipse,minimum size=7mm,inner sep=1pt,font=\small
]
\tikzstyle{cpt}=[
  font=\footnotesize
]

\tikzstyle{graph}=[
  >=stealth,font=\small,auto,scale=1,every node/.style={scale=1}
]
\tikzstyle{node}=[
  draw,circle,inner sep=3pt,fill=white
]

\tikzstyle{bdd}=[
  >=latex, thick, >=stealth, font=\small,auto,scale=0.9,every node/.style={scale=0.9}
]
\tikzstyle{bddnode}=[
  draw,circle,inner sep=0pt,fill=white,minimum size=5.5mm
]

\tikzstyle{highedge}=[
    line width=0.9
]
\tikzstyle{lowedge}=[
    line width=0.9,dotted
]
\tikzstyle{bddterminal}=[
  draw,fill=gray!20,inner sep=2.5pt, font=\small
]

\usepackage{amssymb}
\usepackage{pifont}
\newcommand{\form}[1]{{\mathbf{F}{(#1)}}}
\newcommand{\formp}[1]{{\mathbf{F}'{(#1)}}}
\newcommand{\formpp}[1]{{\mathbf{F}''{(#1)}}}

\newcommand{\indep}{\rotatebox[origin=c]{90}{$\models$}}

\newlength{\pardefault}
\newcommand{\dbracket}[1]{\llbracket {#1} \rrbracket} 
\newcommand{\ts}[0]{\texttt{s}} 
\newcommand{\dippl}[0]{\textsc{dippl}}
\newcommand{\te}[0]{ \texttt{e} }

\newcommand{\true}[0]{\texttt{T}}
\newcommand{\false}[0]{\texttt{F}}
\newcommand{\flip}[0]{\texttt{flip}}
\newcommand{\mods}[0]{\mathtt{Mod}}
\newcommand{\wmc}[0]{\mathtt{WMC}}

\newcommand{\dist}[0]{\mathtt{Dist}}

\newcommand{\sym}[0]{\mathcal{S}}

\def\Pr{\mathop{\rm Pr}\nolimits}

\newcommand{\steven}[1] {}
\newcommand{\todd}[1] {}
\newcommand{\guy}[1] {}
\newcommand{\joe}[1] {}

\begin{document}

\title[]{Symbolic Exact Inference for \\ Discrete Probabilistic Programs}


\author{Steven Holtzen}
\affiliation{
  \institution{UCLA}            
}
\email{sholtzen@cs.ucla.edu}          

\author{Todd Millstein}
\affiliation{
  \institution{UCLA}            
}
\email{todd@cs.ucla.edu}          
\author{Guy Van den Broeck}
\affiliation{
  \institution{UCLA}            
}
\email{guyvdb@cs.ucla.edu}          

\begin{abstract}
  The computational burden of probabilistic inference remains a hurdle for
applying probabilistic programming languages to practical problems of interest.
In this work, we provide a semantic and algorithmic foundation for efficient
exact inference on discrete-valued finite-domain imperative probabilistic
programs. We leverage and generalize efficient inference procedures for Bayesian
networks, which exploit the structure of the network to decompose the inference
task, thereby avoiding full path enumeration. To do this, we first compile
probabilistic programs to a symbolic representation. Then we adapt techniques
from the probabilistic logic programming and artificial intelligence communities
in order to perform inference on the symbolic representation. We formalize our
approach, prove it sound, and experimentally validate it against existing exact
and approximate inference techniques. We show that our inference approach is
competitive with inference procedures specialized for Bayesian networks, thereby
expanding the class of probabilistic programs that can be practically analyzed.
\end{abstract}




\maketitle

\section{Introduction}
When it is computationally feasible, exact probabilistic inference is vastly
preferable to approximation techniques. Exact inference methods are
deterministic and reliable, so they can be trusted for making high-consequence
decisions and do not propagate errors to subsequent analyses. Ideally, one would
use exact inference \emph{whenever possible}, only resorting to approximation
when exact inference strategies become infeasible. Even when approximating
inference, one often performs exact inference in an approximate model. This is
the case for a wide range of approximation schemes, including message
passing~\cite{ChoiDarwiche11,sontag2011introduction}, sampling~\cite{gogate2011samplesearch,FriedmanNIPS18}, and variational
inference~\cite{wainwright2008graphical}.

Existing probabilistic programming systems lag behind state-of-the-art
techniques for performing exact probabilistic inference in other domains such as
graphical models. Fundamentally, inference -- both exact and approximate -- is
theoretically hard \cite{roth1996hardness}. However, exact inference is
routinely performed in practice. This is because many interesting inference
problems have \emph{structure}: there are underlying repetitions and
decompositions that can be exploited to perform inference more efficiently than
the worst case. Existing efficient exact inference procedures -- notably
techniques from the graphical models inference community -- systematically
find and exploit the underlying structure of the problem in order to mitigate
the inherent combinatorial explosion problem of exact probabilistic inference
\cite{boutilier1996context, koller2009probabilistic, Pearl1988PRIS}.

We seek to close the performance gap between exact inference in
discrete graphical models and discrete-valued finite-domain probabilistic
programs. The key idea behind existing state-of-the-art inference procedures in
discrete graphical models is to compile the graphical model into a
representation known as a \emph{weighted Boolean formula} (WBF), which is a
symbolic representation of the joint probability distribution over the graphical
model's random variables. This symbolic representation exposes key structural
elements of the distribution, such as independences between random variables.
Then, inference is performed via a weighted sum of the models of the WBF, a
process known as \emph{weighted model counting} (WMC). This WMC process exploits
the independences present in the WBF and is competitive with state-of-the-art
inference techniques in many domains, such as probabilistic logic programming,
Bayesian networks, and probabilistic databases~\cite{chavira2006compiling,
  Fierens11, VdBFTDB17, VlasselaerAIJ16, chavira2006compiling, ChoiLFU17}.

First we give a motivating example that highlights key properties of our
approach. Then, we describe our symbolic compilation in more detail; the precise
details of our compilation, and its proof of correctness, can be found in the
appendix. Then, we illustrate how to use binary decision diagrams to represent
the probability distribution of a probabilistic program for efficient inference.
Finally, we provide preliminary experimental results illustrating the promise of
this approach on several challenging probabilistic programs.

\section{Exact Symbolic Inference}
\label{sec:motivating_example}

\begin{figure}
  \centering
   \begin{subfigure}[B]{1.0\linewidth}
\begin{lstlisting}[mathescape=true]
$x \sim$flip$_x(0.5)$;
if($x$) { $y \sim$flip$_1(0.6)$ } 
else { $y \sim$flip$_2(0.4)$ };
if($y$) { $z \sim$flip$_3(0.6)$ } 
else { $z \sim$flip$_4$(0.9) }
\end{lstlisting}
     \caption{A simple probabilistic program. The notation $x \sim
       \texttt{flip}_l(\theta)$ denotes drawing a sample from a
       Bernoulli$(\theta)$ distribution and assigning the outcome to the
       variable $x$.  The label $l$ is not actually part of the syntax but is used so we can refer to each \texttt{flip} uniquely.}
     \label{fig:motiv_ex_prog}
   \end{subfigure}\\
   
   \begin{subfigure}[B]{1.0\linewidth}
     \centering
         \begin{tikzpicture}
      
      \def\lvl{22pt}

      \node (fx) at (0,0) [bddnode] {$f_x$};
      \node (xt) at ($(fx) + (-30bp, -\lvl)$) [bddnode] {$x$};
      \node (xf) at ($(fx) + (30bp, -\lvl)$) [bddnode] {$x$};

      \node (falsex) at ($(fx) + (0, -1.5*\lvl)$) [bddterminal] {$\false$};
      \node (fy1) at ($(xt) + (0, -\lvl)$) [bddnode] {$f_1$};
      \node (fy2) at ($(xf) + (0, -\lvl)$) [bddnode] {$f_2$};

      \node (yt) at ($(fy1) + (0, -\lvl)$) [bddnode] {$y$};
      \node (yf) at ($(fy2) + (0, -\lvl)$) [bddnode] {$y$};

      \node (falsey) at ($(falsex) + (0, -2*\lvl)$) [bddterminal] {$\false$};
      \node (fz1) at ($(yt) + (0, -\lvl)$) [bddnode] {$f_3$};
      \node (fz2) at ($(yf) + (0, -\lvl)$) [bddnode] {$f_4$};

      \node (zt) at ($(fz1) + (0, -\lvl)$) [bddnode] {$z$};
      \node (zf) at ($(fz2) + (0, -\lvl)$) [bddnode] {$z$};

      \node (falsez) at ($(zf) + (0, -\lvl)$) [bddterminal] {$\false$};
      \node (truez) at ($(zt) + (0, -\lvl)$) [bddterminal] {$\true$};
    \begin{scope}[on background layer]
      \draw [highedge] (fx) -- (xt);
      \draw [lowedge] (fx) -- (xf);

      \draw [lowedge] (xt) -- (falsex);
      \draw [highedge] (xf) -- (falsex);
      \draw [highedge] (xt) -- (fy1);
      \draw [lowedge] (xf) -- (fy2);

      \draw [highedge] (fy1) -- (yt);
      \draw [lowedge] (fy1) -- (yf);
      \draw [highedge] (fy2) -- (yt);
      \draw [lowedge] (fy2) -- (yf);
      
      \draw [lowedge] (yt) -- (falsey);
      \draw [highedge] (yf) -- (falsey);
      \draw [highedge] (yt) -- (fz1);
      \draw [lowedge] (yf) -- (fz2);

      \draw [highedge] (fz1) -- (zt);
      \draw [lowedge] (fz1) -- (zf);
      \draw [highedge] (fz2) -- (zt);
      \draw [lowedge] (fz2) -- (zf);
      
      \draw [lowedge] (zt) -- (falsez);
      \draw [highedge] (zf) -- (falsez);
      \draw [highedge] (zt) -- (truez);
      \draw [lowedge] (zf) -- (truez);
    \end{scope}
    \end{tikzpicture}
    \caption{A binary decision diagram representing the Boolean formula compiled
      from the program in Figure~\ref{fig:motiv_ex_prog}; a low edge is denoted a dashed
      line, and a high edge is denoted with a solid line. The variables $f_x$, $f_1$,
      $f_2$, $f_3$, and $f_4$ correspond to annotations in Figure~\ref{fig:motiv_ex_prog}.}
    \label{fig:motiv_bdd}
   \end{subfigure}\\

   \caption{Probabilistic program and its symbolic representation.}
   \label{fig:motiv_ex_fig}
   \vspace{-0.5cm}
\end{figure}
In this section we present a motivating example that highlights key elements of
our approach. Figure~\ref{fig:motiv_ex_prog} shows a simple probabilistic
program that encodes a linear Bayesian network, a structure known as a
\emph{Markov chain} \cite{koller2009probabilistic}. In order to perform
inference efficiently on a Markov chain -- or any Bayesian network -- it is
necessary to exploit the \emph{independence structure} of the model. Exploiting
independence is one of the key techniques for efficient graphical model
inference procedures. Markov chains encoded as probabilistic programs have $2^n$
paths, where $n$ is the length of the chain. Thus, inference methods that rely
on exhaustively exploring the paths in a program -- a strategy we refer to as
\emph{path-based} inference methods -- will require exponential time in the
length of the Markov chains; see our experiments in Figure~\ref{fig:exp_chain}.
Path-based inference is currently a common strategy for performing discrete
exact inference in the probabilistic programming
community~\citep{Sankaranarayanan2013, Albarghouthi2017_1, gehr2016psi}.

However, it is well known that Markov chains support linear-time inference in
the length of the chain \cite{koller2009probabilistic}. The reason for this is
that the structure of a Markov chain ensures a strong form of {\em conditional
independence}: each node in the chain depends only on the directly preceding
node in the chain. In the program of Figure~\ref{fig:motiv_ex_prog}, for
example, the probability distribution for $z$ is independent of $x$ \emph{given}
$y$, i.e., if $y$ is fixed to a particular value, then the probability
distribution over $z$ can be computed without considering the distribution over
$x$. Therefore inference can be \emph{factorized}: the probability distribution
for $y$ can be determined as a function of that for $x$, and then the
probability distribution for $z$ can be determined as a function of that for
$y$. More generally, inference for a chain of length $n$ can be reduced to
inference on $n-1$ separate chains, each of length two.

To close this performance gap between Bayesian networks and exact PPL inference,
we leverage and generalize state-of-the-art techniques for Bayesian inference,
which represent the distribution \emph{symbolically}
\cite{Chavira2008,Fierens11}. In this style, the Bayesian network is compiled to
a Boolean function and represented using a binary decision diagram (BDD) or
related data structure \cite{Darwiche2002}. The BDD structure
directly exploits conditional independences -- as well as other forms
of independence -- by caching and re-using duplicate sub-functions during
compilation \cite{Akers1978BinaryDD}.

In this paper we describe an algorithm for compiling a probabilistic program to
a Boolean formula, which can then be represented by a BDD. As an example,
Figure~\ref{fig:motiv_bdd} shows a BDD representation of the program in
Figure~\ref{fig:motiv_ex_prog}. The outcome of each \texttt{flip}$_l(\theta)$
expression in the program is encoded as a Boolean variable labeled $f_l$. A
\emph{model} of the BDD is a truth assignment to all the variables in the BDD
that causes the BDD to return {\tt T}, and each model of the BDD in
Figure~\ref{fig:motiv_bdd} represents a possible execution of the original
program.

The exploitation of the conditional independence structure of the program is
clearly visible in the BDD. For example, any feasible execution in which $y$ is
true has the same sub-function for $z$ --- the subtree rooted at $f_3$ ---
regardless of the value of $x$. The same is true for any feasible execution in
which $y$ is false. More generally, the BDD for a Markov chain has size linear
in the length of the chain, despite the exponential number of possible execution
paths.

To perform inference on this BDD, we first associate a \emph{weight} with each
truth assignment to each variable: the variables $x,y,$ and $z$ are given a
weight of 1 for both the true and false assignments, and the
\texttt{flip}$(\theta)$ variables are given a weight of $\theta$ and $1-\theta$
for their true and false assignments respectively. The Boolean formula together
with these weights is called a \emph{weighted Boolean formula}.

Finally, we can perform inference on the original probabilistic program relative
to a given {\em inference query} (e.g., ``What is the probability that $z$ is
false?'') via \emph{weighted model counting} (WMC). The weight of a model of the
BDD is defined as the product of the weights of each variable assignment in the
model, and the WMC of a set of models is the sum of the weights of the models.
Then the answer to a given inference query $Q$ is simply the WMC of all models
of the BDD that satisfy the query. WMC is a well-studied general-purpose
technique for performing probabilistic inference and is currently the
state-of-the-art technique for inference in discrete Bayesian networks,
probabilistic logic programs, and probabilistic databases \cite{Chavira2008,
Fierens11, VdBFTDB17}. BDDs support linear-time weighted model counting by
performing a single bottom-up pass of the diagram \cite{Darwiche2002}: thus,
we can compile a single BDD for a probabilistic program, which can be used to
exactly answer many inference queries.

\section{The \dippl{} Language}
\label{sec:lang}
Here we formally define the syntax and semantics of our discrete finite-domain imperative
probabilistic programming language \dippl{} language. First we will introduce
and discuss the syntax. Then, we will describe the semantics and its basic
properties. For more details on the semantics, see the appendix.

\subsection{Syntax}

\begin{figure}
  \begin{lstlisting}[mathescape=true]
s ::=
  | s; s
  | x := e
  | x $\sim$ flip($\theta$)
  | if e { s } else { s }
  | observe(e)
  | skip
e :: =
  | $x$
  | $\true$ | $\false$
  | e $\lor$ e
  | e $\land$ e
  | $\neg$ e
\end{lstlisting}
\caption{Syntax of \dippl{}.}
\label{fig:syntax}
\end{figure}


Figure~\ref{fig:syntax} gives the syntax of our probabilistic programming
language \dippl{}. Metavariable $x$ ranges over variable names, and
metavariable $\theta$ ranges over rational numbers in the interval $[0,1]$. All
data is Boolean-valued, and expressions include the usual Boolean operations,
though it is straightforward to extend the language to other finite-domain
datatypes. In addition to the standard loop-free imperative statements, there
are two probablistic statements. The statement {\tt x $\sim$ flip($\theta$)}
samples a value from the Bernoulli distribution defined by parameter $\theta$
(i.e., $\true$ with probability $\theta$ and $\false$ with probability
$1-\theta$). The statement {\tt observe(e)} conditions the current distribution
of the program on the event that $\te$ evaluates to true.

\subsection{Semantics}
The goal of the semantics of any probabilistic programming language is to define
the distribution over which one wishes to perform inference. In this section, we
introduce a denotational semantics that directly produces this distribution of
interest, and it is defined over program states.
A state $\sigma$ is a finite map from variables to Boolean values, and $\Sigma$
is the set of all possible states. 

We define a denotational semantics for \dippl{}, which we call its
\emph{transition semantics} and denote $\dbracket{\cdot}_T$. These semantics are
given in the appendix. The transition
semantics will be the primary semantic object of interest for \dippl{}, and will
directly produce the distribution over which we wish to perform inference.
For some statement $\ts$, the transition semantics is written
$\dbracket{\ts}_T(\sigma' \mid \sigma)$, and it computes the (normalized) conditional
probability upon executing $\ts$ of transitioning to state $\sigma'$. The
transition semantics have the following type signature:
\begin{align*}
  \dbracket{\ts}_T : \Sigma \rightarrow \dist~\Sigma
\end{align*}
where $\dist~\Sigma$ is the set of all probability distributions over $\Sigma$.
For example,
\begin{align*}
          \dbracket{x\sim \flip(\theta)}_T(\sigma' \mid \sigma)
        \triangleq &
                     \begin{cases}
                       \theta \quad& \text{if } \sigma' = \sigma[x \mapsto \true]\\
                       1 - \theta \quad& \text{if } \sigma' = \sigma[x \mapsto \false]\\
                       0 \quad& \text{otherwise}
                     \end{cases}
\end{align*}

Ultimately, our goal during inference is to compute the probability of some
event occurring in the probability distribution defined by the transition
semantics of the program.

\section{Symbolic Compilation for Inference}
\label{sec:symcomp}

Existing approaches to exact inference for imperative PPLs perform {\em path enumeration}:
each execution path is individually analyzed to determine the probability mass
along the path, and the probability masses of all paths are summed. 
As
argued earlier, such approaches are inefficient due to the need to enumerate
complete paths and the inability to take advantage of key properties of the
probability distribution across paths, notably forms of independence.

In this section we present an alternative approach to exact inference for PPLs,
which is inspired by state-of-the-art techniques for exact inference in Bayesian
networks \cite{Chavira2008}. We describe how to compile a probabilistic program
to a \emph{weighted Boolean formula}, which symbolically represents the program
as a relation between input and output states.  
Inference is then reduced to performing a \emph{weighted model count} (WMC) on
this formula, which can be performed efficiently using BDDs and related data
structures.

\subsection{Weighted Model Counting}
Weighted model counting is a well-known general-purpose technique for performing
probabilistic inference in the artificial intelligence and probabilistic logic
programming communities, and it is currently the state-of-the-art technique for
performing inference in certain classes of Bayesian networks and probabilistic
logic programs \cite{Chavira2008,
Fierens11, VdBFTDB17, sang2005performing}. There exist a variety of
general-purpose black-box tools for performing weighted model counting, similar to
satisfiability solvers \cite{OztokD15, OztokD14b, muise2010fast}. \steven{Guy --
  Maybe you can comment here more on the overall impact of WMC?} 

First, we give basic definitions from propositional logic. A \emph{literal} is
either a Boolean variable or its negation. For a
formula $\varphi$ over variables $V$, a sentence $\omega$ is a \emph{model} of
$\varphi$ if it is a conjunction of literals, contains every variable in $V$,
and $\omega \models \varphi$. We denote the set of all models of $\varphi$ as
$\mods(\varphi)$. Now we are ready to define a weighted Boolean formula:

\begin{definition}[Weighted Boolean Formula]
Let $\varphi$ be a Boolean formula, $L$ be the set of all literals for variables
that occur in $\varphi$, and $w : L \rightarrow \mathbb{R}^+$ be a function
that associates a real-valued positive weight with each literal $l \in L$. The
pair $(\varphi, w)$ is a \emph{weighted Boolean formula} (WBF).
\end{definition}

Next, we define the weighted model counting task, which computes a weighted sum
over the models of a weighted Boolean formula:

\begin{definition}[Weighted Model Count]
  Let $(\varphi, w)$ be a weighted Boolean formula. Then, the \emph{weighted
model count} ($\wmc$) of $(\varphi, w)$ is defined as:
\begin{align}
  \wmc(\varphi, w) \triangleq \sum_{\omega \in \mods(\varphi)} \prod_{l \in \omega} w(l)
\end{align}
where the set $l \in \omega$ is the set of all literals in the model $\omega$.
\end{definition}

The process of symbolic compilation associates a \textsc{dippl} program with a
weighted Boolean formula and is described next. 

\subsection{Symbolic Compilation}
We formalize symbolic compilation of a \dippl{}
program to a weighted Boolean formula as a relation denoted $\ts \rightsquigarrow (\varphi,
w)$. The formal rules for this relation are
described in detail in the
appendix; here we describe the important properties of this compilation.
Intuitively, the formula $\varphi$ produced by the compilation represents
the program $\ts$ as a relation between initial states and final states, where
initial states are represented by unprimed Boolean variables $\{x_i\}$ and final
states are represented by primed Boolean variables $\{x_i'\}$. 
These compiled weighted Boolean formulas will have a probabilistic semantics
that allow them to be interpreted as a transition
probability for the original statement.

Our goal is to ultimately give a correspondence between the compiled weighted
Boolean formula and the original denotational semantics of the statement.
First we define the translation of a state $\sigma$ to a logical formula:
\begin{definition}[Boolean state]
  Let $\sigma \!=\! \{(x_1,b_1),\ldots,(x_n,b_n)\}$.  We define the \emph{Boolean
state} $\form{\sigma}$ as $l_1 \wedge \ldots \wedge l_n$ where for each $i$,
$l_i$ is $x_i$ if $\sigma(x_i)=\true$ and $\neg x_i$ if $\sigma(x_i)=\false$.
For convenience, we also define a version that relabels state
variables to their primed versions, $\formp{\sigma} \triangleq \form{\sigma}[x_i
\mapsto x_i']$.
\end{definition}

Now, we formally describe how every compiled weighted Boolean formula can be
interpreted as a conditional probability by computing the appropriate weighted
model count:
\begin{definition}[Compiled semantics] \label{def:wmc-semantics}
  Let $(\varphi, w)$ be a weighted Boolean formula, and let $\sigma$ and $\sigma'$ be states.
  Then, the \emph{transition semantics} of
$(\varphi, w)$ is defined:
  \begin{align}
    \dbracket{(\varphi, w)}_T(\sigma' \mid \sigma) \triangleq
    \frac{\wmc(\varphi \land \form{\sigma} \land \formp{\sigma'}, w)}
    {\wmc(\varphi \land \form{\sigma}, w)} \label{eq:wmc-semantics}
  \end{align}
\end{definition}
Moreover, the transition semantics of Definition~\ref{def:wmc-semantics} allows for more general queries to be phrased as WMC tasks as well. For example, the probability of some event $\alpha$ being true in the output state $\sigma'$ can be computed by replacing $\formp{\sigma'}$ in Equation~\ref{eq:wmc-semantics} by a Boolean formula for $\alpha$.


Finally, we state our correctness theorem, which describes the relation between
the semantics of the compiled WBF to the denotational
semantics of \dippl{}:
\begin{theorem}[Correctness of Compilation Procedure]
  \label{thm:correctness}
  Let $\ts$ be a \dippl{} program, $V$ be the set of all variables in $\ts$,
and  $\ts \rightsquigarrow (\varphi, w)$. Then for all states $\sigma$ and
$\sigma'$ over the variables in $V$, we have:
\begin{align}
    \dbracket{\ts}_T(\sigma' \mid \sigma) = \dbracket{(\varphi, w)}_T(\sigma' \mid \sigma)
      \end{align}
\end{theorem}
\begin{proof}
  A complete proof can be found in Appendix~\ref{sec:app_correctness}.
\end{proof}

The above theorem allows us to perform inference via weighted model
counting on the compiled WBF for a \dippl{} program. See the appendix for
details on this compilation procedure, and proof of its correctness.

\section{Efficient Inference}
\label{sec:efficient_inf}
Inference is theoretically hard \cite{roth1996hardness}. Exploiting the
structure of the problem -- and in particular, exploiting various forms of
independence -- are essential for scalable and practical inference procedures
\cite{boutilier1996context, koller2009probabilistic, Pearl1988PRIS}. In this
section, we will represent a compiled weighted Boolean formula as a binary
decision diagram (BDD). We will show how BDDs implicitly exploit the problem
structure.

\subsection{BDD Representation}
BDDs are a popular choice for representing the set of reachable states in the
symbolic model checking community \citep{Ball2000}. BDDs support a variety of
useful properties which make them suitable for this task: they support an array
of pairwise composition operations, including conjunction, disjunction,
existential quantification and variable relabeling. These composition operators
are \emph{efficient}, i.e. performing them requires time polynomial in the sizes
of the two BDDs that are being composed. 

In addition to supporting efficient compositional operators, BDDs also support a
variety of efficient queries, including satisfiability and weighted model
counting \cite{Darwiche2002}.


\subsection{Exploiting Program Structure}
\begin{figure}
    \begin{subfigure}[b]{\linewidth}
\begin{lstlisting}[mathescape=true]
$x \sim$flip$_1(0.6)$;
$y \sim$flip$_2(0.7)$
\end{lstlisting}
    \caption{A probabilistic program illustrating independence between variables
    $x$ and $y$.}
    \label{fig:indep_ex}
  \end{subfigure}\\
  \begin{subfigure}[b]{\linewidth}
    \centering
    \begin{tikzpicture}
      
    \def\lvl{18pt}
    \node (fx) at (0bp,0bp) [bddnode] {$f_1$};

    \node (xt) at ($(fx) + (-30bp, -\lvl)$) [bddnode] {$x$};
    \node (xf) at ($(fx) + (30bp, -\lvl)$) [bddnode] {$x$};
    
    \node (falsex) at ($(fx) + (0bp, -2*\lvl)$) [bddterminal] {$\false$};

    \node (fy) at ($(falsex) + (0bp, -\lvl)$) [bddnode] {$f_2$};

    \node (yt) at ($(fy) + (-30bp, -\lvl)$) [bddnode] {$y$};
    \node (yf) at ($(fy) + (30bp, -\lvl)$) [bddnode] {$y$};

    \node (truey) at ($(yt) + (0bp, -\lvl)$) [bddterminal] {$\true$};
    \node (falsey) at ($(yf) + (0bp, -\lvl)$) [bddterminal] {$\false$};

    \begin{scope}[on background layer]
      \draw [highedge] (fx) -- (xt);
      \draw [lowedge] (fx) -- (xf);

      \draw [lowedge] (xt) -- (falsex);
      \draw [highedge] (xf) -- (falsex);
      \draw[highedge] (xt) -- (fy);
      \draw[lowedge] (xf) -- (fy);

      \draw [highedge] (fy) -- (yt);
      \draw [lowedge] (fy) -- (yf);

      \draw [lowedge] (yt) -- (falsey);
      \draw [highedge] (yf) -- (falsey);
      \draw[highedge] (yt) -- (truey);
      \draw[lowedge] (yf) -- (truey);

    \end{scope}

    \end{tikzpicture}
    \caption{A BDD representing the logical formula compiled from the program in
      Figure~\ref{fig:indep_ex}. The variables $f_1$ and $f_2$ correspond to the
      \texttt{flip} statements on lines 1 and 2 respectively.
    }
    \label{fig:indep_bdd}
  \end{subfigure}\\

  \begin{subfigure}[b]{\linewidth}
 \begin{lstlisting}[mathescape=true]
$z \sim$flip$_1(0.5)$;
if($z$) {
  $x \sim$flip$_2(0.6)$;
  $y \sim$flip$_3(0.7)$
} else {
  $x \sim$flip$_4(0.4)$;
  y := x
}
\end{lstlisting}
     \caption{A probabilistic program illustrating the context-specific
       independence between $x$ and $y$ given $z = \true$.}
     \label{fig:sens_ex}
   \end{subfigure}\\

   \begin{subfigure}[b]{\linewidth}
     \centering
    \begin{tikzpicture}
      
      \def\lvl{18pt}

    \node (fz) at (0, 0) [bddnode] {$f_1$};

    \node (zt) at ($(fz) + (-40bp, -\lvl)$) [bddnode] {$z$};
    \node (zf) at ($(fz) + (40bp, -\lvl)$) [bddnode] {$z$};
    
    \node (falsez) at ($(fx) + (0bp, -2*\lvl)$) [bddterminal] {$\false$};

    \node (fx2) at ($(zf) + (0, -\lvl)$) [bddnode] {$f_4$};
    \node (xt2) at ($(fx2) + (-20bp, -\lvl)$) [bddnode] {$x$};
    \node (xf2) at ($(fx2) + (20bp, -\lvl)$) [bddnode] {$x$};

    \node (falsex2) at ($(fx2) + (0bp, -2*\lvl)$) [bddterminal] {$\false$};

    \node (yt2) at ($(xt2) + (0, -2*\lvl)$) [bddnode] {$y$};
    \node (yf2) at ($(xf2) + (0, -2*\lvl)$) [bddnode] {$y$};

    \node (rightf) at ($(yf2) + (0, -\lvl)$) [bddterminal] {$\false$};
    \node (rightt) at ($(yt2) + (0, -\lvl)$) [bddterminal] {$\true$};
    
    \node (fx1) at ($(zt) + (0, -\lvl)$) [bddnode] {$f_2$};

    \node (xt) at ($(fx1) + (-20bp, -\lvl)$) [bddnode] {$x$};
    \node (xf) at ($(fx1) + (20bp, -\lvl)$) [bddnode] {$x$};
    
    \node (falsex) at ($(fx1) + (0bp, -2*\lvl)$) [bddterminal] {$\false$};

    \node (fy) at ($(falsex) + (0bp, -\lvl)$) [bddnode] {$f_3$};

    \node (yt) at ($(fy) + (-20bp, -\lvl)$) [bddnode] {$y$};
    \node (yf) at ($(fy) + (20bp, -\lvl)$) [bddnode] {$y$};

    \node (truey) at ($(yt) + (0bp, -\lvl)$) [bddterminal] {$\true$};
    \node (falsey) at ($(yf) + (0bp, -\lvl)$) [bddterminal] {$\false$};

    \begin{scope}[on background layer]
      \draw [highedge] (fz) -- (zt);
      \draw [lowedge] (fz) -- (zf);

      \draw [highedge] (zt) -- (fx1);
      \draw [lowedge] (zt) -- (falsez);
      \draw [lowedge] (zf) -- (fx2);
      \draw [highedge] (zf) -- (falsez);

      \draw [highedge] (fx2) -- (xt2);
      \draw [lowedge] (fx2) -- (xf2);

      \draw [highedge] (xt2) -- (yt2);
      \draw [lowedge] (xt2) -- (falsex2);
      \draw [lowedge] (xf2) -- (yf2);
      \draw [highedge] (xf2) -- (falsex2);

      \draw [highedge] (yt2) -- (rightt);
      \draw [lowedge] (yt2) -- (rightf);
      \draw [lowedge] (yf2) -- (rightt);
      \draw [highedge] (yf2) -- (rightf);
      
      \draw [highedge] (fx1) -- (xt);
      \draw [lowedge] (fx1) -- (xf);

      \draw [lowedge] (xt) -- (falsex);
      \draw [highedge] (xf) -- (falsex);
      \draw[highedge] (xt) -- (fy);
      \draw[lowedge] (xf) -- (fy);

      \draw [highedge] (fy) -- (yt);
      \draw [lowedge] (fy) -- (yf);

      \draw [lowedge] (yt) -- (falsey);
      \draw [highedge] (yf) -- (falsey);
      \draw[highedge] (yt) -- (truey);
      \draw[lowedge] (yf) -- (truey);

    \end{scope}
    \end{tikzpicture}

   \caption{A BDD representing the logical formula compiled from the program in
      Figure~\ref{fig:sens_ex}. The variables $f_1,$ $f_2$, $f_3$, and $f_4$ correspond to the
      annotated \texttt{flip} statements.}
     \label{fig:sens_bdd}
   \end{subfigure}
   \caption{Example \dippl{} programs and the BDDs each of them 
     compile to. This compilation assumes that the
     initial state is the true BDD.}
   \label{fig:bdds}
\end{figure}

Compilation to BDDs -- and related representations -- is currently the
state-of-the-art approach to inference in certain kinds of discrete Bayesian
networks, probabilistic logic programs, and probabilistic databases
\cite{Chavira2008, Fierens11, VdBFTDB17}. The fundamental reason is that BDDs
exploit \emph{duplicate sub-functions}: if there is a sub-function that is
constructed more than once in the symbolic compilation, that duplicate sub-function
is cached and re-used. This sub-function deduplication is critical for efficient
inference. In this section, we explore how BDDs exploit specific properties of
the program and discuss when a program will have a small BDD.

\paragraph{Independence}
Exploiting independence is essential for efficient inference and is the
backbone of existing state-of-the-art inference algorithms. There are three
kinds of independence structure which we seek to exploit. The first is the
strongest form:
\begin{definition}[Independence]
Let $\Pr(X,Y)$ be a joint probability distribution over sets of random variables $X$ and
$Y$. Then, we say that $X$ is \emph{independent} of $Y$, written $X \indep Y$,
if $\Pr(X,Y) = \Pr(X) \times \Pr(Y)$. In this case, we say that this
distribution \emph{factorizes} over the variables $X$ and $Y$.
\end{definition}
Figure~\ref{fig:indep_ex} shows a probabilistic program with two independent
random variables $x$ and $y$. The corresponding BDD generated in
Figure~\ref{fig:indep_bdd} exploits the independence between the variables $x$
and $y$. In particular, we see that node $f_2$ does not depend on the particular
value of $x$. Thus, the BDD \emph{factorizes} the distribution over $x$ and $y$.
As a consequence, the size of the BDD grows linearly with the number of
independent random variables.

\paragraph{Conditional independence}
The next form of independence we consider is \emph{conditional independence}:
\begin{definition}[Conditional independence]
  Let $\Pr(X,Y,Z)$ be a joint probability distribution over sets of random variables
  $X,Y$, and $Z$. Then, we say $X$ is independent of $Z$ \emph{given} $Y$,
  written $X \indep Z \mid Y$, if $\Pr(X,Z \mid Y) = \Pr(X \mid Y) \times \Pr(Z
  \mid Y)$.
\end{definition}
Figure~\ref{fig:motiv_ex_fig} gave an example probabilistic program that
exhibits conditional independence. In this program, the variables $x$ and $z$
are correlated unless $y$ is fixed to a particular value: thus, $x$ and $z$ are
conditionally independent given $y$. Figure~\ref{fig:motiv_bdd} shows how this
conditional independence is exploited by the BDD; thus, Markov chains have
BDD representations that are linear in size to the length of chain.

Conditional independence is exploited by specialized inference
algorithms for Bayesian networks like the join-tree algorithm \citep{koller2009probabilistic}.
However, conditional independence is not exploited by path-based -- or
enumerative -- probabilistic program inference procedures, such as the method
utilized by Psi \cite{Gehr2016}.

\paragraph{Context-specific independence}
The final form of independence we will discuss is context-specific independence.
Context-specific independence is a weakening of conditional independence that occurs
when two sets of random variables are independent only when a third set of
variables all take on a particular \emph{value} \cite{boutilier1996context}:

\begin{definition}[Context-specific independence]
  Consider a joint probability distribution $\Pr(X,Y,Z)$ over sets of random variables
$X,Y$, and $Z$, and let $c$ be an assignment to variables in $Z$. Then, we
say $X$ is contextually independent of $Y$ \emph{given} $Z = c$, written $X
\indep Y \mid Z = c$, if $\Pr(X,Y \mid Z = c) = \Pr(X \mid Z = c) \times \Pr(Y
\mid Z=c)$.
\end{definition}
An example program that exhibits context-specific independence is show in Figure~\ref{fig:sens_ex}.
The variables $x$ and $y$ are
correlated if $z = \false$ or if $z$ is unknown, but they are independent if $z=
\true$. Thus, $x$ is independent of $y$ given $z = \true$.

The equivalent BDD generated in Figure~\ref{fig:sens_bdd} exploits the conditional
independence of $x$ and $y$ given $z = \true$ by first branching on the value of $z$,
and then representing the configurations of $x$ and $y$ as two sub-functions. 
Note here that the variable order
of the BDD is relevant. The BDD generated in Figure~\ref{fig:sens_bdd} exploits
the context-specific independence of $x$ and $y$ given $z=\true$ by
representing $x$ and $y$ in a factorized manner when $z = \true$. Note how the
sub-function when $z=\true$ is isomorphic to Figure~\ref{fig:indep_bdd}.

In general, exploiting context-specific independence is challenging and is not
directly supported in typical Bayesian network inference algorithms such as the
join-tree algorithm. Context-specific independence is often present when there
is some amount of \emph{determinism}, and exploiting context-specific
independence was one of the original motivations for the development of WMC for Bayesian networks \cite{sang2005performing,Chavira2008}. Probabilistic
programs are very often partially deterministic; thus, we believe exploiting
context-specific independence is essential for practical efficient inference in
this domain. To our knowledge, no existing imperative or functional PPL
inference system currently exploits context-specific independence.

\section{Implementation \& Experiments}

\begin{figure}
  \centering
      \begin{tikzpicture}
        \begin{axis}[
          symbolic x coords={Alarm, Two Coins, Noisy Or, Grass},
          xtick = {Alarm, Two Coins, Noisy Or, Grass},
          x tick label style={rotate=25,},
          ylabel=Time (ms),
          ymax = 1200,
          enlargelimits=0.15, 
          legend style={at={(0.03,0.7)},anchor=west},
          ybar,
          width=8cm,
          height=5cm,
          tick label style={font=\footnotesize},
          ]

          \addplot[fill=white] coordinates
          {(Alarm, 17)
            (Two Coins, 9)
            (Noisy Or, 296)
            (Grass, 277)
          };

          \addplot[fill=gray] coordinates
          {(Alarm, 15)
            (Two Coins, 5)
            (Noisy Or, 1116)
            (Grass, 254)
          };

          \addplot[pattern=north west lines] coordinates
          {
            (Alarm, 173)
            (Grass, 336)
            (Noisy Or, 250)
            (Two Coins, 108)
          };

          \legend{Symbolic, Psi, R2}
        \end{axis}
\end{tikzpicture}
\caption{Performance comparison for discrete inference between Psi
  \cite{Gehr2016}, R2 \cite{Nori2014}, and symbolic inference. The $y$-axis gives
  the time in milliseconds to perform inference.}
\label{fig:baselines}
\end{figure}

\label{sec:experiments}
\begin{figure}
  \centering
  \begin{tikzpicture}
	\begin{axis}[
		height=5cm,
		width=8cm,
		grid=major,
    xlabel={Length of Markov Chain},
    ylabel={Time (s)},
    cycle list name=black white,
	]

	\addplot coordinates {
(3.00  ,0.256)
(5.00  ,0.262) 
(7.00  ,0.284) 
(10.00 ,0.316) 
(15.00 ,0.359) 
(18.00 ,0.386) 
(20.00 ,0.408) 
(22.00 ,0.425) 
(23.00 ,0.438)
(24.00 ,0.447)
(30.00 ,0.501)
(40.00 ,0.596)
(50.00 ,0.694)
(60.00 ,0.791)
(70.00 ,0.895)
(90.00 ,1.1)
(110.00,1.31)
(130.00,1.55)
(150.00,1.74)
	};
	\addlegendentry{Symbolic (This Work)}

  \addplot coordinates {

		(3,0.8)
    (5, 0.28)
    (7, 2.37)
    (10, 25.27)
    (11, 80.23)
	};
	\addlegendentry{Psi}

  \addplot[mark=diamond] coordinates {

		(3,0.0)
    (10,0.0)
    (12,1.138)
    (13, 3.89)
    (14, 52.94)
    (15, 188.09)
	};
	\addlegendentry{Storm}

  \addplot[mark=triangle] coordinates {
(3.00 , 1.12)
(5.00 , 1.13)
(7.00 , 1.13)
(10.00,	1.19)
(15.00,	1.33)
(18.00,	3.00)
(20.00,	6.81)
(22.00,	27.7)
(23.00,	50.1)
(24.00,	103.50)
  };
  
	\addlegendentry{WebPPL}
	\end{axis}
\end{tikzpicture}
\caption{Experimental comparison between
  techniques for performing exact inference on a Markov chain.}
\label{fig:exp_chain}
\vspace{-0.4cm}
\end{figure}

In this section we experimentally validate the effectiveness of our symbolic
compilation procedure for performing inference on \dippl{} programs. We directly
implemented the compilation procedure described in Section~\ref{sec:symcomp} in
Scala. We used the JavaBDD library in order create and manipulate binary
decision diagrams~\citep{javabdd}.

\subsection{Experiments}
Our goal is to
validate that it is a viable technique for performing inference in practice and
performs favorably in comparison with existing exact (and approximate) inference
techniques.

First, we discuss a collection of simple baseline inference tasks to
demonstrate that our symbolic compilation is competitive with Psi \cite{Gehr2016}, R2
\cite{Nori2014}, and the Storm probabilistic model checker
\citep{dehnert2017storm}. Then, we elaborate on the motivating example from
Section~\ref{sec:motivating_example} and clearly demonstrate how our symbolic
approach can exploit conditional independence to scale to large Markov models.
Next, we show how our technique can achieve performance that is competitive
with specialized Bayesian network inference techniques. Finally, we demonstrate
how our symbolic compilation can exploit context-specific independence to
perform inference on a synthetic grids dataset. All experiments were conducted
on a 2.3GHz Intel i5 processor with 16GB of RAM.

\subsubsection{Baselines}
In Figure~\ref{fig:baselines} we compared our technique against Psi
\cite{Gehr2016} and R2 \cite{Nori2014} on the collection of all discrete
probabilistic programs that they were both evaluated on. Psi\footnote{We used
Psi version \texttt{52b31ba}.} is an exact inference compilation technique, so
its performance can be directly compared against our performance. R2 is an
approximate inference engine and cannot produce exact inference results. The
timings reported for R2 are the time it took R2 to produce an approximation that
is within 3\% of the exact answer\footnote{Our performance figures for R2 are
excerpted from \citet{Gehr2016}. We were not able to run R2 to perform our own
experiments due to inability to access the required version of Visual Studio.}.

The code for each of the models -- Alarm, Two Coins, Noisy Or, and Grass -- was
extracted from the source code found in the R2 and Psi source code repositories
and then translated to \dippl{}. These baseline experiments show that our
symbolic technique is competitive with existing methods on well-known example
models. However, these examples are too small to demonstrate the benefits of
symbolic inference: each example is less than 25 lines. In subsequent
sections, we will demonstrate the power of symbolic inference by exploiting
independence structure in much larger discrete models.

\subsubsection{Markov Chain} Section~\ref{sec:motivating_example} discussed Markov
chains and demonstrated that a compact BDD can be compiled that exploits the
conditional independence of the network. In particular, a Markov chain of length
$n$ can be compiled to a linear-sized BDD in $n$.

Figure~\ref{fig:exp_chain} shows how two exact probabilistic programming
inference tools compare against our symbolic inference technique for inference
on Markov chains. WebPPL \cite{wingate2013automated} and Psi \cite{Gehr2016}
rely on enumerative concrete exact inference, which is exponential in the length
of the chain.
To compare against Storm, we compiled these
models directly into discrete-time Markov chains. As the length of the Markov
chain grows, the
size of the encoded discrete-time Markov chain grows exponentially.
Symbolic inference exploits the conditional independence of each
variable in the chain, and is thus linear time in the length of the chain.

\subsubsection{Bayesian Network Encodings}
\begin{table}
  \centering
  \begin{tabular}{lrrr}
    \toprule
    \textbf{Model} & \textbf{Us (s)} & \textbf{BN Time (s) \citep{Chavira2008}} & \textbf{Size of BDD} \\
    \midrule
    Alarm \citep{Chavira2008} & 1.872 & 0.21 & 52k \\
    Halfinder	& 12.652 & 1.37 & 157k \\
    Hepar2 & 7.834 & 0.28~\citep{dal2018parallel} & 139k	 \\
    pathfinder & 62.034 & 14.94 & 392k \\
    \bottomrule
  \end{tabular} 
  \vspace{2mm}
  \caption{Experimental results for Bayesian networks encoded as probabilistic
    programs. We report the time it took to perform exact inference in seconds for
    our method compared against the Bayesian network inference algorithm from
    \citet{Chavira2008}, labeled as ``BN Time''. In addition, we report the final
    size of our compiled BDD.}
  \label{tab:bn}
\end{table}

In this section we demonstrate the power of our symbolic representation by
performing exact inference on Bayesian networks encoded as probabilistic
programs. We compared the performance of our symbolic compilation procedure
against an exact inference procedure for Bayesian networks \citep{Chavira2008}.
Each of these Bayesian networks is from \citet{Chavira2008}\footnote{The
networks can also be found at \url{http://www.bnlearn.com/bnrepository}}.
Table~\ref{tab:bn} shows the experimental results: our symbolic approach is
competitive with specialized Bayesian network inference. 

The goal of these experiments is to benchmark the extent to which one sacrifices
efficient inference for a more flexible modeling framework; \textsc{Ace} is at
an inherent advantage in this comparison for two main reasons. First, our
inference algorithm is compositional, while \textsc{Ace} considers the whole
Bayesian network at once. This gives \textsc{Ace} an advantage on this
benchmark. \textsc{Ace} compiles Bayesian networks to \emph{d-DNNFs}, which is a family
of circuits that are
not efficiently composable, but are faster to compile than
BDDs~\citep{Darwiche2002}. Our technique compiles to BDDs, which are slower to
compile than d-DNNFs, but support a compositional line-by-line compilation
procedure. Second, Bayesian
networks are in some sense a worst-case probabilistic program, since they have
no interesting program structure beyond the graph structure that \textsc{Ace}
already exploits.

These Bayesian networks are not necessarily Boolean valued: they may contain
multi-valued nodes. For instance, the Alarm network has three values that the
StrokeVolume variable may take. We encode these multi-valued nodes as Boolean program variables
using a one-hot encoding in a style similar to \citet{sang2005performing}. The
generated \dippl{} files are quite large: the pathfinder program has over ten
thousand lines of code. Furthermore, neither ProbLog \cite{Fierens2013,
de2007problog} nor Psi could perform inference within 300 seconds on the alarm
example, the smallest of the above examples, thus demonstrating the power of our
encoding over probabilistic logic programs and enumerative inference on this example.

\subsubsection{Grids}
\begin{figure}
  \centering
  \begin{tikzpicture}
	\begin{axis}[
		height=5cm,
		width=8cm,
    xtick={4,5,6},
		grid=major,
    xlabel={Size of Grid},
    ylabel={Time (s)},
    legend style={at={(0.03,0.7)},anchor=west},
    cycle list name=black white,
	]

	\addplot coordinates {
    (4, 0.868)
    (5, 25.428)
	};
	\addlegendentry{0\%}

  \addplot coordinates {
		(4,0.5)
    (5,0.978)
    (6,85)
	};
	\addlegendentry{50\%}

  \addplot[mark=triangle] coordinates {
    (4, 0.454)
    (5, 0.72)
    (6, 3.893)
  };
	\addlegendentry{90\%}
	\end{axis}
\end{tikzpicture}
\caption{Experiment evaluating the effects of determinism on compiling an
  encoding of a grid Bayesian network. The $n$\% result means that there is
  $n$\% determinism present. Time was cut off at a max of 300 seconds.}
\label{fig:exp_grid}
\vspace{-0.4cm}
\end{figure}
This experiment showcases how our method exploits context-specific independence
to perform inference more efficiently in the presence of determinism. Grids were
originally introduced by \citet{sang2005performing} to demonstrate the
effectiveness of exploiting determinism during Bayesian network inference. A 3-grid
is a Boolean-valued Bayesian network arranged in a three by three grid:

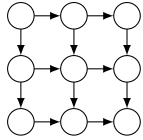
\begin{center}
 \begin{tikzpicture}
  \draw (0,0) circle (5bp);
  \draw (20bp,0) circle (5bp);
  \draw (40bp,0) circle (5bp);

  \draw (0,20bp) circle (5bp);
  \draw (20bp,20bp) circle (5bp);
  \draw (40bp,20bp) circle (5bp);

  \draw (0,40bp) circle (5bp);
  \draw (20bp,40bp) circle (5bp);
  \draw (40bp,40bp) circle (5bp);

  \draw[->] (5bp, 0) -- (15bp, 0);
  \draw[->] (25bp,0) -- (35bp, 0);
  \draw[->] (5bp, 20bp) -- (15bp, 20bp);
  \draw[->] (25bp,20bp) -- (35bp, 20bp);
  \draw[->] (5bp, 40bp) -- (15bp, 40bp);
  \draw[->] (25bp,40bp) -- (35bp, 40bp);

  \draw[->] (0bp,  15bp) -- (0bp, 5bp) ;
  \draw[->] (0bp, 35bp) -- (0bp, 25bp) ;
  \draw[->] (20bp,  15bp) -- (20bp, 5bp) ;
  \draw[->] (20bp, 35bp) --  (20bp, 25bp) ;
  \draw[->] (40bp,  15bp) -- (40bp, 5bp) ;
  \draw[->] (40bp, 35bp) --  (40bp, 25bp) ;
\end{tikzpicture}
\end{center}

For these experiments we encoded grid Bayesian networks into probabilistic
programs. Grids are typically hard inference challenges even for specialized
Bayesian network inference algorithms. However, in the presence of determinism,
the grid inference task can become vastly easier. A grid is
\emph{$n$\%}-deterministic if $n\%$ of the \texttt{flip}s in the program are
replaced with assignments to constants. Figure~\ref{fig:exp_grid} shows how our
symbolic compilation exploits the context-specific independence induced by the
determinism of the program in order to perform inference more efficiently.

\section{Related Work}
\label{sec:related_work}
First we discuss two closely related individual works on exact inference for
PPLs; then we discuss larger categories of related work.

\citet{Claret2013} compiles imperative probabilistic progams to algebraic
decision diagrams (ADDs) via a form of data-flow analysis~\cite{Claret2013}. This
approach is fundamentally different from our approach, as the ADD cannot
represent the distribution in a factorized way. An ADD must contain the
probability of each model of the Boolean formula as a leaf node. Thus, it cannot
exploit the independence structure required to compactly represent joint
probability distributions with independence structure efficiently. 

Also closely related is the work of \citet{pfeffer2018structured}, which seeks
to decompose the probabilistic program inference task at specific program points
where the distribution is known to factorize due to conditional independence.
This line of work only considers conditional independence --- not
context-specific independence --- and requires hand-annotated program constructs
in order to expose and exploit the independences.

\paragraph{Path-based Program Inference} Many techniques for performing
inference in current probabilistic programming languages are enumerative
or \emph{path-based}: they perform inference by integrating or approximating the
probability mass along each path of the probabilistic program \citep{Gehr2016,
wingate2013automated, Sankaranarayanan2013, Albarghouthi2017_1}. The complexity
of inference for path-based inference algorithms scales with the number of paths
through the program. The main weakness with these inference strategies is that
they cannot exploit common structure across paths -- such as independence --
and thus scale poorly on examples with many paths.

\paragraph{Probabilistic Logic Programs}
Most prior work on exact inference for probabilistic programs was developed for
probabilistic logic programs~\cite{de2007problog, Riguzzi2011TPLP,
Fierens2013, Vlasselaer2015, RenkensAAAI14}. Similar to our work, these techniques compile
a probablistic logic program into a weighted Boolean formula and utilize
state-of-the-art WMC solvers to compile the WBF into a representation that
supports efficient WMC evaluation, such as a binary decision diagram
(BDD)~\cite{Bryant86}, sentential decision diagram (SDD)~\cite{darwiche2011sdd},
or d-DNNF circuit~\cite{Darwiche2002}. Currently, WMC-based inference remains
the state-of-the-art inference strategy for probabilistic logic programs. These
techniques are not directly applicable to imperative probabilistic programs such
as \dippl{} due to the presence of sequencing, arbitrary observation, and other
imperative programming constructs.

\paragraph{Probabilistic Model Checkers}
Probabilistic model checkers such as Storm~\citep{dehnert2017storm} and
Prism~\citep{Kwiatkowska2011} can be used to perform Bayesian inference on
probabilistic systems. These methods work by compiling programs to a
representation such as a discrete-time Markov chain or Markov decision process,
for which there exist well-known inference strategies. These
representations allow probabilistic model checkers to reason about loops and
non-termination. In comparison with this work, probabilistic model checkers
suffer from a state-space explosion similar to path-based inference methods due
to the fact that they devote a node to each possible configuration of variables
in the program.

\paragraph{Compilation-based PPLs}
There exists a large number of PPLs that perform inference by converting the
program into a probabilistic graphical model~\cite{pfeffer2001ibal,
pfeffer2009figaro, McCallum2009, InferNET14}, assuming a fixed set of random
variables. There are two primary shortcomings of these techniques in relation to
ours. The first is that these techniques cannot exploit the context-specific
independence present in the program structure, since the topology of the graph
-- either a Bayesian network or factor graph -- does not make this information
explicit. Second, these techniques restrict the space of programs to those that
can be compiled. Thus they require constraints on the space of programs, such as
requiring a statically-determined number of variables, or requiring that loops
can be statically unrolled. Currently, we have similar constraints in that our
compilation technique cannot handle unbounded loops, that we hope to address in
future work.

\section{Conclusion \& Future Work}
\label{sec:conclusion}
In conclusion, we developed a semantics and symbolic compilation procedure for
exact inference in a discrete imperative probabilistic programming language
called \dippl{}. In doing so, we have drawn connections among the
probabilistic logic programming, symbolic model checking, and artificial
intelligence communities. We theoretically proved our symbolic compilation
procedure correct and experimentally validated it against existing
probabilistic systems. Finally, we showed that our method is competitive with
state-of-the-art Bayesian network inference tasks, showing that our compilation
procedures scales to large complex probability models.

We anticipate much future work in this direction.
First, we plan to extend our symbolic compilation procedure to handle
richer classes of programs. For instance, we would like to support almost-surely
terminating loops and procedures, as well as enrich the class of datatypes
supported by the language. Second, we would like to quantify precisely
the complexity of inference for discrete probabilistic programs. The graphical
models community has metrics such as \emph{tree-width} that provide precise
notions of the complexity of inference; we believe such notions may exist for
probabilistic programs as well \cite{koller2009probabilistic,
  darwiche2009modeling}. Finally, we anticipate that techniques from the
symbolic model checking community -- such as Bebop \cite{Ball2000} -- may be
applicable here, and applying these techniques is also promising future work.

\begin{acks}
 This work is partially supported by \grantsponsor{SP784}{National Science
Foundation}{} grants \grantnum{SP784}{IIS-1657613},
\grantnum{SP784}{IIS-1633857}, and
\grantnum{SP784}{CCF-1837129}; \grantsponsor{SP242}{DARPA}{}~XAI grant
\grantnum{SP242}{N66001-17-2-4032}, \grantsponsor{SP4126}{NEC Research}{}, a
gift from \grantsponsor{SP1450}{Intel}{}, and a gift from Facebook Research. The authors would like to thank Joe
Qian for assistance with the development of the language semantics and its
properties.
\end{acks}

\bibliography{bib}


\begin{thebibliography}{46}


\ifx \showCODEN    \undefined \def \showCODEN     #1{\unskip}     \fi
\ifx \showDOI      \undefined \def \showDOI       #1{#1}\fi
\ifx \showISBNx    \undefined \def \showISBNx     #1{\unskip}     \fi
\ifx \showISBNxiii \undefined \def \showISBNxiii  #1{\unskip}     \fi
\ifx \showISSN     \undefined \def \showISSN      #1{\unskip}     \fi
\ifx \showLCCN     \undefined \def \showLCCN      #1{\unskip}     \fi
\ifx \shownote     \undefined \def \shownote      #1{#1}          \fi
\ifx \showarticletitle \undefined \def \showarticletitle #1{#1}   \fi
\ifx \showURL      \undefined \def \showURL       {\relax}        \fi
\providecommand\bibfield[2]{#2}
\providecommand\bibinfo[2]{#2}
\providecommand\natexlab[1]{#1}
\providecommand\showeprint[2][]{arXiv:#2}

\bibitem[\protect\citeauthoryear{Akers}{Akers}{1978}]%
        {Akers1978BinaryDD}
\bibfield{author}{\bibinfo{person}{Sheldon~B. Akers}.}
  \bibinfo{year}{1978}\natexlab{}.
\newblock \showarticletitle{Binary Decision Diagrams}.
\newblock \bibinfo{journal}{\emph{IEEE Trans. Comput.}}  \bibinfo{volume}{C-27}
  (\bibinfo{year}{1978}), \bibinfo{pages}{509--516}.
\newblock


\bibitem[\protect\citeauthoryear{Albarghouthi, D'Antoni, Drews, and
  Nori}{Albarghouthi et~al\mbox{.}}{2017}]%
        {Albarghouthi2017_1}
\bibfield{author}{\bibinfo{person}{Aws Albarghouthi}, \bibinfo{person}{Loris
  D'Antoni}, \bibinfo{person}{Samuel Drews}, {and} \bibinfo{person}{Aditya~V.
  Nori}.} \bibinfo{year}{2017}\natexlab{}.
\newblock \showarticletitle{FairSquare: Probabilistic Verification of Program
  Fairness}.
\newblock \bibinfo{journal}{\emph{Proc. ACM Program. Lang.}}
  \bibinfo{volume}{1}, \bibinfo{number}{OOPSLA}, Article
  \bibinfo{articleno}{80} (\bibinfo{date}{Oct.} \bibinfo{year}{2017}),
  \bibinfo{numpages}{30}~pages.
\newblock
\showISSN{2475-1421}
\urldef\tempurl%
\url{https://doi.org/10.1145/3133904}
\showDOI{\tempurl}


\bibitem[\protect\citeauthoryear{Ball and Rajamani}{Ball and Rajamani}{2000}]%
        {Ball2000}
\bibfield{author}{\bibinfo{person}{Thomas Ball} {and}
  \bibinfo{person}{Sriram~K. Rajamani}.} \bibinfo{year}{2000}\natexlab{}.
\newblock \showarticletitle{Bebop: A Symbolic Model Checker for Boolean
  Programs}. In \bibinfo{booktitle}{\emph{SPIN Model Checking and Software
  Verification}}. \bibinfo{pages}{113--130}.
\newblock


\bibitem[\protect\citeauthoryear{Boutilier, Friedman, Goldszmidt, and
  Koller}{Boutilier et~al\mbox{.}}{1996}]%
        {boutilier1996context}
\bibfield{author}{\bibinfo{person}{Craig Boutilier}, \bibinfo{person}{Nir
  Friedman}, \bibinfo{person}{Moises Goldszmidt}, {and} \bibinfo{person}{Daphne
  Koller}.} \bibinfo{year}{1996}\natexlab{}.
\newblock \showarticletitle{Context-specific independence in Bayesian
  networks}. In \bibinfo{booktitle}{\emph{Proceedings of the Twelfth
  international conference on Uncertainty in artificial intelligence}}. Morgan
  Kaufmann Publishers Inc., \bibinfo{pages}{115--123}.
\newblock


\bibitem[\protect\citeauthoryear{Bryant}{Bryant}{1986}]%
        {Bryant86}
\bibfield{author}{\bibinfo{person}{R. Bryant}.}
  \bibinfo{year}{1986}\natexlab{}.
\newblock \showarticletitle{Graph-based algorithms for {Boolean} function
  manipulation}.
\newblock \bibinfo{journal}{\emph{IEEE TC}}  \bibinfo{volume}{C-35}
  (\bibinfo{year}{1986}), \bibinfo{pages}{677--691}.
\newblock


\bibitem[\protect\citeauthoryear{Chavira and Darwiche}{Chavira and
  Darwiche}{2008}]%
        {Chavira2008}
\bibfield{author}{\bibinfo{person}{Mark Chavira} {and} \bibinfo{person}{Adnan
  Darwiche}.} \bibinfo{year}{2008}\natexlab{}.
\newblock \showarticletitle{On Probabilistic Inference by Weighted Model
  Counting}.
\newblock \bibinfo{journal}{\emph{J. Artificial Intelligence}}
  \bibinfo{volume}{172}, \bibinfo{number}{6-7} (\bibinfo{date}{April}
  \bibinfo{year}{2008}), \bibinfo{pages}{772--799}.
\newblock
\showISSN{0004-3702}
\urldef\tempurl%
\url{https://doi.org/10.1016/j.artint.2007.11.002}
\showDOI{\tempurl}


\bibitem[\protect\citeauthoryear{Chavira, Darwiche, and Jaeger}{Chavira
  et~al\mbox{.}}{2006}]%
        {chavira2006compiling}
\bibfield{author}{\bibinfo{person}{Mark Chavira}, \bibinfo{person}{Adnan
  Darwiche}, {and} \bibinfo{person}{Manfred Jaeger}.}
  \bibinfo{year}{2006}\natexlab{}.
\newblock \showarticletitle{{Compiling relational Bayesian networks for exact
  inference}}.
\newblock \bibinfo{journal}{\emph{International Journal of Approximate
  Reasoning}} \bibinfo{volume}{42}, \bibinfo{number}{1} (\bibinfo{year}{2006}),
  \bibinfo{pages}{4--20}.
\newblock


\bibitem[\protect\citeauthoryear{Choi and Darwiche}{Choi and Darwiche}{2011}]%
        {ChoiDarwiche11}
\bibfield{author}{\bibinfo{person}{Arthur Choi} {and} \bibinfo{person}{Adnan
  Darwiche}.} \bibinfo{year}{2011}\natexlab{}.
\newblock \showarticletitle{Relax, Compensate and then Recover}.
\newblock In \bibinfo{booktitle}{\emph{New Frontiers in Artificial
  Intelligence}}, \bibfield{editor}{\bibinfo{person}{Takashi Onada},
  \bibinfo{person}{Daisuke Bekki}, {and} \bibinfo{person}{Eric McCready}}
  (Eds.). \bibinfo{series}{Lecture Notes in Computer Science},
  Vol.~\bibinfo{volume}{6797}. \bibinfo{publisher}{Springer Berlin /
  Heidelberg}, \bibinfo{pages}{167--180}.
\newblock


\bibitem[\protect\citeauthoryear{Choi, Darwiche, and Van~den Broeck}{Choi
  et~al\mbox{.}}{2017}]%
        {ChoiLFU17}
\bibfield{author}{\bibinfo{person}{YooJung Choi}, \bibinfo{person}{Adnan
  Darwiche}, {and} \bibinfo{person}{Guy Van~den Broeck}.}
  \bibinfo{year}{2017}\natexlab{}.
\newblock \showarticletitle{Optimal Feature Selection for Decision Robustness
  in Bayesian Networks}. In \bibinfo{booktitle}{\emph{IJCAI 2017 Workshop on
  Logical Foundations for Uncertainty and Machine Learning}}.
\newblock
\urldef\tempurl%
\url{http://starai.cs.ucla.edu/papers/ChoiLFU17.pdf}
\showURL{%
\tempurl}


\bibitem[\protect\citeauthoryear{Claret, Rajamani, Nori, Gordon, and
  Borgstr{\"{o}}m}{Claret et~al\mbox{.}}{2013}]%
        {Claret2013}
\bibfield{author}{\bibinfo{person}{Guillaume Claret},
  \bibinfo{person}{Sriram~K. Rajamani}, \bibinfo{person}{Aditya~V. Nori},
  \bibinfo{person}{Andrew~D. Gordon}, {and} \bibinfo{person}{Johannes
  Borgstr{\"{o}}m}.} \bibinfo{year}{2013}\natexlab{}.
\newblock \showarticletitle{{Bayesian inference using data flow analysis}}.
\newblock \bibinfo{journal}{\emph{Proceedings of the 2013 9th Joint Meeting on
  Foundations of Software Engineering - ESEC/FSE 2013}} (\bibinfo{year}{2013}),
  \bibinfo{pages}{92}.
\newblock
\showISBNx{9781450322379}
\urldef\tempurl%
\url{https://doi.org/10.1145/2491411.2491423}
\showDOI{\tempurl}


\bibitem[\protect\citeauthoryear{Dal, Laarman, and Lucas}{Dal
  et~al\mbox{.}}{2018}]%
        {dal2018parallel}
\bibfield{author}{\bibinfo{person}{Giso~H Dal}, \bibinfo{person}{Alfons~W
  Laarman}, {and} \bibinfo{person}{Peter~JF Lucas}.}
  \bibinfo{year}{2018}\natexlab{}.
\newblock \showarticletitle{Parallel Probabilistic Inference by Weighted Model
  Counting}. In \bibinfo{booktitle}{\emph{International Conference on
  Probabilistic Graphical Models}}. \bibinfo{pages}{97--108}.
\newblock


\bibitem[\protect\citeauthoryear{Darwiche}{Darwiche}{2009}]%
        {darwiche2009modeling}
\bibfield{author}{\bibinfo{person}{A. Darwiche}.}
  \bibinfo{year}{2009}\natexlab{}.
\newblock \bibinfo{booktitle}{\emph{Modeling and Reasoning with {Bayesian}
  Networks}}.
\newblock \bibinfo{publisher}{Cambridge University Press}.
\newblock
\showISBNx{9780521884389}
\showLCCN{2008044605}


\bibitem[\protect\citeauthoryear{Darwiche}{Darwiche}{2011}]%
        {darwiche2011sdd}
\bibfield{author}{\bibinfo{person}{Adnan Darwiche}.}
  \bibinfo{year}{2011}\natexlab{}.
\newblock \showarticletitle{SDD: A new canonical representation of
  propositional knowledge bases}. In \bibinfo{booktitle}{\emph{IJCAI
  Proceedings-International Joint Conference on Artificial Intelligence}}.
  \bibinfo{pages}{819}.
\newblock


\bibitem[\protect\citeauthoryear{Darwiche and Marquis}{Darwiche and
  Marquis}{2002}]%
        {Darwiche2002}
\bibfield{author}{\bibinfo{person}{A. Darwiche} {and} \bibinfo{person}{P.
  Marquis}.} \bibinfo{year}{2002}\natexlab{}.
\newblock \showarticletitle{A Knowledge Compilation Map}.
\newblock \bibinfo{journal}{\emph{Journal of Artificial Intelligence Research}}
   \bibinfo{volume}{17} (\bibinfo{year}{2002}), \bibinfo{pages}{229--264}.
\newblock


\bibitem[\protect\citeauthoryear{De~Raedt, Kimmig, and Toivonen}{De~Raedt
  et~al\mbox{.}}{2007}]%
        {de2007problog}
\bibfield{author}{\bibinfo{person}{Luc De~Raedt}, \bibinfo{person}{Angelika
  Kimmig}, {and} \bibinfo{person}{Hannu Toivonen}.}
  \bibinfo{year}{2007}\natexlab{}.
\newblock \showarticletitle{ProbLog: A Probabilistic Prolog and Its Application
  in Link Discovery.}. In \bibinfo{booktitle}{\emph{Proceedings of IJCAI}},
  Vol.~\bibinfo{volume}{7}. \bibinfo{pages}{2462--2467}.
\newblock


\bibitem[\protect\citeauthoryear{Dehnert, Junges, Katoen, and Volk}{Dehnert
  et~al\mbox{.}}{2017}]%
        {dehnert2017storm}
\bibfield{author}{\bibinfo{person}{Christian Dehnert},
  \bibinfo{person}{Sebastian Junges}, \bibinfo{person}{Joost-Pieter Katoen},
  {and} \bibinfo{person}{Matthias Volk}.} \bibinfo{year}{2017}\natexlab{}.
\newblock \showarticletitle{A storm is coming: A modern probabilistic model
  checker}. In \bibinfo{booktitle}{\emph{International Conference on Computer
  Aided Verification}}. Springer, \bibinfo{pages}{592--600}.
\newblock


\bibitem[\protect\citeauthoryear{Fierens, Van~den Broeck, Renkens, Shterionov,
  Gutmann, Thon, Janssens, and De~Raedt}{Fierens et~al\mbox{.}}{2013}]%
        {Fierens2013}
\bibfield{author}{\bibinfo{person}{Daan Fierens}, \bibinfo{person}{Guy Van~den
  Broeck}, \bibinfo{person}{Joris Renkens}, \bibinfo{person}{Dimitar
  Shterionov}, \bibinfo{person}{Bernd Gutmann}, \bibinfo{person}{Ingo Thon},
  \bibinfo{person}{Gerda Janssens}, {and} \bibinfo{person}{Luc De~Raedt}.}
  \bibinfo{year}{2013}\natexlab{}.
\newblock \showarticletitle{Inference and learning in probabilistic logic
  programs using weighted Boolean formulas}.
\newblock \bibinfo{journal}{\emph{J. Theory and Practice of Logic Programming}}
   \bibinfo{volume}{15(3)} (\bibinfo{year}{2013}), \bibinfo{pages}{358 -- 401}.
\newblock


\bibitem[\protect\citeauthoryear{Fierens, Van~den Broeck, Thon, Gutmann, and
  De~Raedt}{Fierens et~al\mbox{.}}{2011}]%
        {Fierens11}
\bibfield{author}{\bibinfo{person}{Daan Fierens}, \bibinfo{person}{Guy Van~den
  Broeck}, \bibinfo{person}{Ingo Thon}, \bibinfo{person}{Bernd Gutmann}, {and}
  \bibinfo{person}{Luc De~Raedt}.} \bibinfo{year}{2011}\natexlab{}.
\newblock \showarticletitle{Inference in probabilistic logic programs using
  weighted {CNF}'s}. In \bibinfo{booktitle}{\emph{Proceedings of UAI}}.
  \bibinfo{pages}{211--220}.
\newblock


\bibitem[\protect\citeauthoryear{Friedman and Van~den Broeck}{Friedman and
  Van~den Broeck}{2018}]%
        {FriedmanNIPS18}
\bibfield{author}{\bibinfo{person}{Tal Friedman} {and} \bibinfo{person}{Guy
  Van~den Broeck}.} \bibinfo{year}{2018}\natexlab{}.
\newblock \showarticletitle{Approximate Knowledge Compilation by Online
  Collapsed Importance Sampling}. In \bibinfo{booktitle}{\emph{Advances in
  Neural Information Processing Systems 31 (NIPS)}}.
\newblock


\bibitem[\protect\citeauthoryear{Gehr, Misailovic, and Vechev}{Gehr
  et~al\mbox{.}}{2016a}]%
        {gehr2016psi}
\bibfield{author}{\bibinfo{person}{Timon Gehr}, \bibinfo{person}{Sasa
  Misailovic}, {and} \bibinfo{person}{Martin Vechev}.}
  \bibinfo{year}{2016}\natexlab{a}.
\newblock \showarticletitle{Psi: Exact symbolic inference for probabilistic
  programs}. In \bibinfo{booktitle}{\emph{International Conference on Computer
  Aided Verification}}. Springer, \bibinfo{pages}{62--83}.
\newblock


\bibitem[\protect\citeauthoryear{Gehr, Misailovic, and Vechev}{Gehr
  et~al\mbox{.}}{2016b}]%
        {Gehr2016}
\bibfield{author}{\bibinfo{person}{Timon Gehr}, \bibinfo{person}{Sasa
  Misailovic}, {and} \bibinfo{person}{Martin Vechev}.}
  \bibinfo{year}{2016}\natexlab{b}.
\newblock \showarticletitle{PSI: Exact symbolic inference for probabilistic
  programs}.
\newblock \bibinfo{journal}{\emph{Proc.~of ESOP/ETAPS}}  \bibinfo{volume}{9779}
  (\bibinfo{year}{2016}), \bibinfo{pages}{62--83}.
\newblock
\showISBNx{9783319415277}
\showISSN{16113349}


\bibitem[\protect\citeauthoryear{Gogate and Dechter}{Gogate and
  Dechter}{2011}]%
        {gogate2011samplesearch}
\bibfield{author}{\bibinfo{person}{V. Gogate} {and} \bibinfo{person}{R.
  Dechter}.} \bibinfo{year}{2011}\natexlab{}.
\newblock \showarticletitle{SampleSearch: Importance sampling in presence of
  determinism}.
\newblock \bibinfo{journal}{\emph{Artificial Intelligence}}
  \bibinfo{volume}{175}, \bibinfo{number}{2} (\bibinfo{year}{2011}),
  \bibinfo{pages}{694--729}.
\newblock


\bibitem[\protect\citeauthoryear{Koller and Friedman}{Koller and
  Friedman}{2009}]%
        {koller2009probabilistic}
\bibfield{author}{\bibinfo{person}{D. Koller} {and} \bibinfo{person}{N.
  Friedman}.} \bibinfo{year}{2009}\natexlab{}.
\newblock \bibinfo{booktitle}{\emph{Probabilistic graphical models: principles
  and techniques}}.
\newblock \bibinfo{publisher}{MIT press}.
\newblock


\bibitem[\protect\citeauthoryear{Kwiatkowska, Norman, and Parker}{Kwiatkowska
  et~al\mbox{.}}{2011}]%
        {Kwiatkowska2011}
\bibfield{author}{\bibinfo{person}{Marta Kwiatkowska}, \bibinfo{person}{Gethin
  Norman}, {and} \bibinfo{person}{David Parker}.}
  \bibinfo{year}{2011}\natexlab{}.
\newblock \showarticletitle{PRISM 4.0: Verification of Probabilistic Real-time
  Systems}. In \bibinfo{booktitle}{\emph{Proceedings of the 23rd International
  Conference on Computer Aided Verification}}
  \emph{(\bibinfo{series}{CAV'11})}. \bibinfo{publisher}{Springer-Verlag},
  \bibinfo{address}{Berlin, Heidelberg}, \bibinfo{pages}{585--591}.
\newblock
\showISBNx{978-3-642-22109-5}


\bibitem[\protect\citeauthoryear{McCallum, Schultz, and Singh}{McCallum
  et~al\mbox{.}}{2009}]%
        {McCallum2009}
\bibfield{author}{\bibinfo{person}{a McCallum}, \bibinfo{person}{K Schultz},
  {and} \bibinfo{person}{S Singh}.} \bibinfo{year}{2009}\natexlab{}.
\newblock \showarticletitle{{Factorie: Probabilistic programming via
  imperatively defined factor graphs}}.
\newblock \bibinfo{journal}{\emph{Proc.~of NIPS}}  \bibinfo{volume}{22}
  (\bibinfo{year}{2009}), \bibinfo{pages}{1249--1257}.
\newblock
\showISBNx{9781615679119}
\showISSN{03643417}


\bibitem[\protect\citeauthoryear{Minka, Winn, Guiver, Webster, Zaykov, Yangel,
  Spengler, and Bronskill}{Minka et~al\mbox{.}}{2014}]%
        {InferNET14}
\bibfield{author}{\bibinfo{person}{T. Minka}, \bibinfo{person}{J.M. Winn},
  \bibinfo{person}{J.P. Guiver}, \bibinfo{person}{S. Webster},
  \bibinfo{person}{Y. Zaykov}, \bibinfo{person}{B. Yangel}, \bibinfo{person}{A.
  Spengler}, {and} \bibinfo{person}{J. Bronskill}.}
  \bibinfo{year}{2014}\natexlab{}.
\newblock \bibinfo{title}{{Infer.NET 2.6}}.
\newblock
\newblock
\newblock
\shownote{Microsoft Research Cambridge.
  http://research.microsoft.com/infernet.}


\bibitem[\protect\citeauthoryear{Muise, McIlraith, Beck, and Hsu}{Muise
  et~al\mbox{.}}{2010}]%
        {muise2010fast}
\bibfield{author}{\bibinfo{person}{C. Muise}, \bibinfo{person}{S. McIlraith},
  \bibinfo{person}{J.C. Beck}, {and} \bibinfo{person}{E. Hsu}.}
  \bibinfo{year}{2010}\natexlab{}.
\newblock \showarticletitle{Fast {d-DNNF} Compilation with {sharpSAT}}. In
  \bibinfo{booktitle}{\emph{Workshops at the Twenty-Fourth AAAI Conference on
  Artificial Intelligence}}.
\newblock


\bibitem[\protect\citeauthoryear{Nori, Rajamani, and Samuel}{Nori
  et~al\mbox{.}}{2014}]%
        {Nori2014}
\bibfield{author}{\bibinfo{person}{Aditya~V Nori}, \bibinfo{person}{Sriram~K
  Rajamani}, {and} \bibinfo{person}{Selva Samuel}.}
  \bibinfo{year}{2014}\natexlab{}.
\newblock \showarticletitle{{R2 : An Efficient MCMC Sampler for Probabilistic
  Programs}}.
\newblock \bibinfo{journal}{\emph{Aaai}} (\bibinfo{year}{2014}),
  \bibinfo{pages}{2476--2482}.
\newblock
\showISBNx{9781577356806}


\bibitem[\protect\citeauthoryear{Oztok and Darwiche}{Oztok and
  Darwiche}{2014}]%
        {OztokD14b}
\bibfield{author}{\bibinfo{person}{Umut Oztok} {and} \bibinfo{person}{Adnan
  Darwiche}.} \bibinfo{year}{2014}\natexlab{}.
\newblock \showarticletitle{On Compiling CNF into Decision-DNNF}. In
  \bibinfo{booktitle}{\emph{Proceedings of the 20th International Conference on
  Principles and Practice of Constraint Programming (CP)}}.
  \bibinfo{pages}{42--57}.
\newblock


\bibitem[\protect\citeauthoryear{Oztok and Darwiche}{Oztok and
  Darwiche}{2015}]%
        {OztokD15}
\bibfield{author}{\bibinfo{person}{U. Oztok} {and} \bibinfo{person}{A.
  Darwiche}.} \bibinfo{year}{2015}\natexlab{}.
\newblock \showarticletitle{A Top-Down Compiler for Sentential Decision
  Diagrams}. In \bibinfo{booktitle}{\emph{Proceedings of the 24th International
  Joint Conference on Artificial Intelligence (IJCAI)}}.
\newblock


\bibitem[\protect\citeauthoryear{Pearl}{Pearl}{1988}]%
        {Pearl1988PRIS}
\bibfield{author}{\bibinfo{person}{Judea Pearl}.}
  \bibinfo{year}{1988}\natexlab{}.
\newblock \bibinfo{booktitle}{\emph{{Probabilistic Reasoning in Intelligent
  Systems: Networks of Plausible Inference}}}.
\newblock \bibinfo{publisher}{Morgan Kaufmann}.
\newblock


\bibitem[\protect\citeauthoryear{Pfeffer}{Pfeffer}{2001}]%
        {pfeffer2001ibal}
\bibfield{author}{\bibinfo{person}{A. Pfeffer}.}
  \bibinfo{year}{2001}\natexlab{}.
\newblock \showarticletitle{{IBAL}: A probabilistic rational programming
  language}. In \bibinfo{booktitle}{\emph{International Joint Conference on
  Artificial Intelligence}}, Vol.~\bibinfo{volume}{17}.
  \bibinfo{pages}{733--740}.
\newblock


\bibitem[\protect\citeauthoryear{Pfeffer}{Pfeffer}{2009}]%
        {pfeffer2009figaro}
\bibfield{author}{\bibinfo{person}{Avi Pfeffer}.}
  \bibinfo{year}{2009}\natexlab{}.
\newblock \showarticletitle{Figaro: An object-oriented probabilistic
  programming language}.
\newblock \bibinfo{journal}{\emph{Charles River Analytics Technical Report}}
  \bibinfo{volume}{137} (\bibinfo{year}{2009}).
\newblock


\bibitem[\protect\citeauthoryear{Pfeffer, Ruttenberg, Kretschmer, and
  OConnor}{Pfeffer et~al\mbox{.}}{2018}]%
        {pfeffer2018structured}
\bibfield{author}{\bibinfo{person}{Avi Pfeffer}, \bibinfo{person}{Brian
  Ruttenberg}, \bibinfo{person}{William Kretschmer}, {and}
  \bibinfo{person}{Alison OConnor}.} \bibinfo{year}{2018}\natexlab{}.
\newblock \showarticletitle{Structured Factored Inference for Probabilistic
  Programming}. In \bibinfo{booktitle}{\emph{International Conference on
  Artificial Intelligence and Statistics}}. \bibinfo{pages}{1224--1232}.
\newblock


\bibitem[\protect\citeauthoryear{Renkens, Kimmig, Van~den Broeck, and
  De~Raedt}{Renkens et~al\mbox{.}}{2014}]%
        {RenkensAAAI14}
\bibfield{author}{\bibinfo{person}{Joris Renkens}, \bibinfo{person}{Angelika
  Kimmig}, \bibinfo{person}{Guy Van~den Broeck}, {and} \bibinfo{person}{Luc
  De~Raedt}.} \bibinfo{year}{2014}\natexlab{}.
\newblock \showarticletitle{Explanation-based approximate weighted model
  counting for probabilistic logics}. In \bibinfo{booktitle}{\emph{Proceedings
  of the 28th AAAI Conference on Artificial Intelligence, AAAI}}.
\newblock
\urldef\tempurl%
\url{http://starai.cs.ucla.edu/papers/RenkensAAAI14.pdf}
\showURL{%
\tempurl}


\bibitem[\protect\citeauthoryear{Riguzzi and Swift}{Riguzzi and Swift}{2011}]%
        {Riguzzi2011TPLP}
\bibfield{author}{\bibinfo{person}{Fabrizio Riguzzi} {and}
  \bibinfo{person}{Terrance Swift}.} \bibinfo{year}{2011}\natexlab{}.
\newblock \showarticletitle{{The {PITA} System: Tabling and Answer Subsumption
  for Reasoning under Uncertainty}}.
\newblock \bibinfo{journal}{\emph{Theory and Practice of Logic Programming}}
  \bibinfo{volume}{11}, \bibinfo{number}{4--5} (\bibinfo{year}{2011}),
  \bibinfo{pages}{433--449}.
\newblock


\bibitem[\protect\citeauthoryear{Roth}{Roth}{1996}]%
        {roth1996hardness}
\bibfield{author}{\bibinfo{person}{Dan Roth}.} \bibinfo{year}{1996}\natexlab{}.
\newblock \showarticletitle{On the hardness of approximate reasoning}.
\newblock \bibinfo{journal}{\emph{Artificial Intelligence}}
  \bibinfo{volume}{82}, \bibinfo{number}{1} (\bibinfo{year}{1996}),
  \bibinfo{pages}{273--302}.
\newblock


\bibitem[\protect\citeauthoryear{Sang, Beame, and Kautz}{Sang
  et~al\mbox{.}}{2005}]%
        {sang2005performing}
\bibfield{author}{\bibinfo{person}{Tian Sang}, \bibinfo{person}{Paul Beame},
  {and} \bibinfo{person}{Henry~A Kautz}.} \bibinfo{year}{2005}\natexlab{}.
\newblock \showarticletitle{Performing Bayesian inference by weighted model
  counting}. In \bibinfo{booktitle}{\emph{AAAI}}, Vol.~\bibinfo{volume}{5}.
  \bibinfo{pages}{475--481}.
\newblock


\bibitem[\protect\citeauthoryear{Sankaranarayanan, Chakarov, and
  Gulwani}{Sankaranarayanan et~al\mbox{.}}{2013}]%
        {Sankaranarayanan2013}
\bibfield{author}{\bibinfo{person}{Sriram Sankaranarayanan},
  \bibinfo{person}{Aleksandar Chakarov}, {and} \bibinfo{person}{Sumit
  Gulwani}.} \bibinfo{year}{2013}\natexlab{}.
\newblock \showarticletitle{Static Analysis for Probabilistic Programs:
  Inferring Whole Program Properties from Finitely Many Paths}.
\newblock \bibinfo{journal}{\emph{SIGPLAN Not.}} \bibinfo{volume}{48},
  \bibinfo{number}{6} (\bibinfo{date}{June} \bibinfo{year}{2013}),
  \bibinfo{pages}{447--458}.
\newblock
\showISSN{0362-1340}
\urldef\tempurl%
\url{https://doi.org/10.1145/2499370.2462179}
\showDOI{\tempurl}


\bibitem[\protect\citeauthoryear{Sontag, Globerson, and Jaakkola}{Sontag
  et~al\mbox{.}}{2011}]%
        {sontag2011introduction}
\bibfield{author}{\bibinfo{person}{David Sontag}, \bibinfo{person}{Amir
  Globerson}, {and} \bibinfo{person}{Tommi Jaakkola}.}
  \bibinfo{year}{2011}\natexlab{}.
\newblock \showarticletitle{Introduction to dual composition for inference}.
\newblock In \bibinfo{booktitle}{\emph{Optimization for Machine Learning}}.
  \bibinfo{publisher}{MIT Press}.
\newblock


\bibitem[\protect\citeauthoryear{Van~den Broeck and Suciu}{Van~den Broeck and
  Suciu}{2017}]%
        {VdBFTDB17}
\bibfield{author}{\bibinfo{person}{Guy Van~den Broeck} {and}
  \bibinfo{person}{Dan Suciu}.} \bibinfo{year}{2017}\natexlab{}.
\newblock \bibinfo{booktitle}{\emph{Query Processing on Probabilistic Data: A
  Survey}}.
\newblock \bibinfo{publisher}{Now Publishers}.
\newblock
\urldef\tempurl%
\url{https://doi.org/10.1561/1900000052}
\showDOI{\tempurl}


\bibitem[\protect\citeauthoryear{Vlasselaer, Meert, Van~den Broeck, and
  De~Raedt}{Vlasselaer et~al\mbox{.}}{2016}]%
        {VlasselaerAIJ16}
\bibfield{author}{\bibinfo{person}{Jonas Vlasselaer}, \bibinfo{person}{Wannes
  Meert}, \bibinfo{person}{Guy Van~den Broeck}, {and} \bibinfo{person}{Luc
  De~Raedt}.} \bibinfo{year}{2016}\natexlab{}.
\newblock \showarticletitle{Exploiting Local and Repeated Structure in Dynamic
  Bayesian Networks}.
\newblock \bibinfo{journal}{\emph{Artificial Intelligence}}
  \bibinfo{volume}{232} (\bibinfo{date}{March} \bibinfo{year}{2016}),
  \bibinfo{pages}{43 -- 53}.
\newblock
\showISSN{0004-3702}
\urldef\tempurl%
\url{https://doi.org/10.1016/j.artint.2015.12.001}
\showDOI{\tempurl}


\bibitem[\protect\citeauthoryear{Vlasselaer, {Van den Broeck}, Kimmig, Meert,
  and {De Raedt}}{Vlasselaer et~al\mbox{.}}{2015}]%
        {Vlasselaer2015}
\bibfield{author}{\bibinfo{person}{Jonas Vlasselaer}, \bibinfo{person}{Guy {Van
  den Broeck}}, \bibinfo{person}{Angelika Kimmig}, \bibinfo{person}{Wannes
  Meert}, {and} \bibinfo{person}{Luc {De Raedt}}.}
  \bibinfo{year}{2015}\natexlab{}.
\newblock \showarticletitle{{Anytime inference in probabilistic logic programs
  with {T}p-compilation}}. In \bibinfo{booktitle}{\emph{{Proceedings of 24th
  International Joint Conference on Artificial Intelligence (IJCAI)}}}.
\newblock


\bibitem[\protect\citeauthoryear{Wainwright, Jordan, et~al\mbox{.}}{Wainwright
  et~al\mbox{.}}{2008}]%
        {wainwright2008graphical}
\bibfield{author}{\bibinfo{person}{Martin~J Wainwright},
  \bibinfo{person}{Michael~I Jordan}, {et~al\mbox{.}}}
  \bibinfo{year}{2008}\natexlab{}.
\newblock \showarticletitle{Graphical models, exponential families, and
  variational inference}.
\newblock \bibinfo{journal}{\emph{Foundations and Trends{\textregistered} in
  Machine Learning}} \bibinfo{volume}{1}, \bibinfo{number}{1--2}
  (\bibinfo{year}{2008}), \bibinfo{pages}{1--305}.
\newblock


\bibitem[\protect\citeauthoryear{Whaley}{Whaley}{2007}]%
        {javabdd}
\bibfield{author}{\bibinfo{person}{John Whaley}.}
  \bibinfo{year}{2007}\natexlab{}.
\newblock \bibinfo{title}{JavaBDD}.
\newblock \bibinfo{howpublished}{\url{http://javabdd.sourceforge.net}}.
\newblock


\bibitem[\protect\citeauthoryear{Wingate and Weber}{Wingate and Weber}{2013}]%
        {wingate2013automated}
\bibfield{author}{\bibinfo{person}{David Wingate} {and}
  \bibinfo{person}{Theophane Weber}.} \bibinfo{year}{2013}\natexlab{}.
\newblock \showarticletitle{Automated variational inference in probabilistic
  programming}.
\newblock \bibinfo{journal}{\emph{arXiv preprint arXiv:1301.1299}}
  (\bibinfo{year}{2013}).
\newblock


\end{thebibliography}
\appendix
\onecolumn
\section{Semantics of \dippl{}}
\label{sec:denot}
\begin{figure}
  \centering
  \begin{subfigure}[B]{1.0\linewidth}
    \begin{mdframed}
      \begin{align*}
        \begin{split}
          \dbracket{\texttt{skip}}_T(\sigma' \mid \sigma) \triangleq&
          \begin{cases}
            1 \quad& \text{if } \sigma' = \sigma \\
            0 \quad& \text{otherwise}
          \end{cases}
        \end{split}\\
        \vspace{0.2cm}
        \dbracket{x\sim \flip(\theta)}_T(\sigma' \mid \sigma)
        \triangleq &
                     \begin{cases}
                       \theta \quad& \text{if } \sigma' = \sigma[x \mapsto \true]\\
                       1 - \theta \quad& \text{if } \sigma' = \sigma[x \mapsto \false]\\
                       0 \quad& \text{otherwise}
                     \end{cases}\\
        \vspace{0.2cm}
        \dbracket{x := \te}_T(\sigma'  \mid \sigma) \triangleq&
                                                                \begin{cases}
                                                                  1 \quad& \text{if } \sigma' = \sigma[x \mapsto \dbracket{\te}(\sigma)] \\
                                                                  0 \quad& \text{otherwise}
                                                                \end{cases}\\
        \dbracket{\texttt{observe}(\te)}_T(\sigma' \mid \sigma) \triangleq&
                                                                            \begin{cases}
                                                                              1 & \text{if } \sigma' = \sigma \text{ and } \dbracket{\te}(\sigma) = \true\\
                                                                              0 & \text{otherwise}
                                                                            \end{cases}
      \end{align*}
      \begin{align*}
        \dbracket{\ts_1; \ts_2}_T&(\sigma' \mid \sigma) \triangleq\\
                                 &\frac{\sum_{\tau \in \Sigma} \dbracket{\ts_1}_T(\tau \mid \sigma) \times
                                   \dbracket{\ts_2}_T(\sigma' \mid \tau) \times \dbracket{\ts_2}_A(\tau)}
                                   {\sum_{\tau \in \Sigma} \dbracket{\ts_1}_T(\tau \mid \sigma) \times \dbracket{\ts_2}_A(\tau) }
      \end{align*}
      \begin{align*}
        \begin{split}
          \dbracket{\texttt{if}~\te~\{\ts_1\}~\texttt{else}~\{\ts_2\}}_T&(\sigma' \mid \sigma)  \triangleq \\
          &\begin{cases}
            \dbracket{\ts_1}_T(\sigma' \mid \sigma) ~~& \text{if } \dbracket{\te}(\sigma) = \true \\
            \dbracket{\ts_2}_T(\sigma' \mid \sigma) ~~& \text{if } \dbracket{\te}(\sigma) = \false \\
          \end{cases}
        \end{split} 
      \end{align*}
    \end{mdframed}
    \caption{Transition semantics of \dippl{}.
      $\dbracket{\ts}_T(\sigma' \mid \sigma)$ gives the conditional
      probability upon executing $\ts$ of transitioning to state $\sigma'$ \emph{given} that
      the start state is $\sigma$ and no \texttt{observe} statements in $\ts$ are violated.
      If every execution path violates an observation,
      $\dbracket{\ts}_T(\sigma' \mid \sigma) = \bot$. }
    \label{fig:transition}
  \end{subfigure}\\

  \begin{subfigure}[B]{1.0\linewidth}
    \begin{mdframed}
      \begin{align*}
        \dbracket{\texttt{skip}(\te)}_A(\sigma) \triangleq& 1 \\
        \dbracket{x \sim \texttt{flip}(\theta)}_A(\sigma) \triangleq& 1 \\
        \dbracket{x := \te}_A(\sigma) \triangleq& 1 \\
        \dbracket{\texttt{observe}(\te)}_A(\sigma) \triangleq&
                                                               \begin{cases}
                                                                 1 \quad& \text{if } \dbracket{\te}(\sigma) = \true\\
                                                                 0 \quad& \text{otherwise}
                                                               \end{cases}\\
        \dbracket{\ts_1; \ts_2}_A(\sigma) \triangleq
        \dbracket{\ts_1}_A(\sigma) \times
        \sum_{\tau \in \Sigma}&\left(  \dbracket{\ts_1}_T  (\tau \mid \sigma) \times \dbracket{\ts_2}_A(\tau) \right) \\
        \dbracket{\texttt{if}~\te~\{\ts_1\}~\texttt{else}~\{\ts_2\}}_A(\sigma)
        \triangleq&
                    \begin{cases}
                      \dbracket{\ts_1}_A(\sigma) & \text{if } \dbracket{\te}(\sigma) = \true\\
                      \dbracket{\ts_2}_A(\sigma) & \text{otherwise}
                    \end{cases}
      \end{align*}
    \end{mdframed}
    \caption{Accepting semantics of \dippl{}. $\dbracket{\ts}_A(\sigma)$ gives
      the probability that no observations are violated by executing $\ts$
      beginning in state~$\sigma$.}
    \label{fig:accepting}
  \end{subfigure}   \caption{Semantics of \dippl{}.}
\end{figure}

The goal of the semantics of any probabilistic programming language is to define
the distribution over which one wishes to perform inference. In this section, we
introduce a denotational semantics that directly produces this distribution of
interest, and it is defined over program states.
A state $\sigma$ is a finite map from variables to Boolean values, and $\Sigma$
is the set of all possible states. We will be interested in probability
distributions on $\Sigma$, defined formally as follows:
\begin{definition}[Discrete probability distribution]
  Let $\Omega$ be a set called the \emph{sample space}. Then, a \emph{discrete
probability distribution} on $\Omega$ is a function $\Pr : 2^{\omega}
\rightarrow [0,1]$ such that (1) $\Pr(\Omega) = 1$; (2) for any $\omega \in
\Omega$, $\Pr(\omega) \ge 0$; (3) for any countable set of
disjoint elements $\{A_i\}$ where $A_i \subseteq 2^{\Omega}$, we have that
$\Pr(\bigcup_i \{A_i\}) = \sum_i \Pr(A_i)$.
\end{definition}

We denote the set of all possible discrete probability distributions with $\Sigma$
as the sample space as $\dist~\Sigma$. We add a special element to
$\dist~\Sigma$, denoted $\bot$, which is the function that assigns a probability
of zero to all states (this will be necessary to represent situations where an
{\tt observe}d expression is false).

We define a denotational semantics for \dippl{}, which we call its
\emph{transition semantics} and denote $\dbracket{\cdot}_T$. These semantics are
summarized in Figure~\ref{fig:transition}. The transition
semantics will be the primary semantic object of interest for \dippl{}, and will
directly produce the distribution over which we wish to perform inference.
For some statement $\ts$, the transition semantics is written
$\dbracket{\ts}_T(\sigma' \mid \sigma)$, and it computes the conditional
probability upon executing $\ts$ of transitioning to state $\sigma'$ \emph{given} that
the start state is $\sigma$ and no \texttt{observe} statements are violated. The
transition semantics have the following type signature:
\begin{align*}
  \dbracket{\ts}_T : \Sigma \rightarrow \dist~\Sigma
\end{align*}

The transition semantics of \dippl{} is shown in Figure~\ref{fig:transition}.
The semantics of {\tt skip}, assignment, and conditionals are straightforward.
The semantics of sampling from a Bernoulli distribution is analogous to that for
assignment, except that there are two possible output states depending on the
value that was sampled. An {\tt observe} statement has no effect if the
associated expression is true in $\sigma$; otherwise the semantics has the
effect of mapping $\sigma$ to the special $\bot$ distribution.

\paragraph{The Role of Observe in Sequencing}
The transition semantics of \dippl{} require that each statement be interpreted
as a conditional probability. Ideally, we would like this conditional
probability to be sufficient to describe the semantics of compositions. Perhaps
surprisingly, the conditional probability distribution of transitioning from one
state to another alone is insufficient for capturing the behavior of
compositions in the presence of observations. We will illustrate this principle
with an example. Consider the
following two \dippl{} statements:
\begin{equation*}
\mathit{bar}_1 = \left\{\begin{minipage}{.25\textwidth}
\begin{lstlisting}[mathescape=true,aboveskip=0mm,frame=none,numbers=none,
  xleftmargin=3mm]
if($x$) { $y \sim$ flip($1/4$) } 
else { $y \sim$ flip($1/2$) }
\end{lstlisting}
\end{minipage} \right\},
\end{equation*}
\begin{equation*}
\mathit{bar}_2 = \left\{\begin{minipage}{.25\textwidth}
\begin{lstlisting}[mathescape=true,aboveskip=0mm,frame=none,numbers=none,
  xleftmargin=3mm]
$y \sim$ flip($1/2$);
observe($x \lor y$);
if($y$) { $y \sim$ flip($1/2$) } 
else { $y$ := $\false$ }
\end{lstlisting}
\end{minipage} \right\}.
\end{equation*}

Both statements represent exactly the same \emph{conditional} probability distribution from input to output states:
\begin{align*}
\begin{split}
  &\dbracket{\mathit{bar}_1}_T(\sigma' \mid \sigma)
   = \dbracket{\mathit{bar}_2}_T(\sigma' \mid \sigma) \\
  & \qquad =  \begin{cases}
    1/2 \quad& \text{if } x[\sigma] = x[\sigma'] = \false,  \\
    1/4 \quad& \text{if } x[\sigma] = x[\sigma'] = \true \text{ and } y[\sigma']=\true, \\
    3/4 \quad& \text{if } x[\sigma] = x[\sigma'] = \true \text{ and } y[\sigma']=\false, \\
    0 \quad& \text{otherwise.}
  \end{cases}
\end{split}
\end{align*}
This is easy to see for $\mathit{bar}_1$, which encodes these probabilities
directly. For $\mathit{bar}_2$, intuitively, when $y=\true$ in the output, both
{\tt flip} statements must return $\true$, which
happens with probability $1/4$. When $x=\false$ in the input, $\mathit{bar}_2$
uses an {\tt observe} statement to disallow executions where the first {\tt
flip} returned $\false$. Given this observation, the {\em then} branch is always taken,
so output $y=\true$ has probability $1/2$.

Because the purpose of probabilistic programming is often to represent a
conditional probability distribution, one is easily fooled into believing that
these programs are equivalent. This is not the case: $\mathit{bar}_1$ and
$\mathit{bar}_2$ behave differently when sequenced with other statements. For
example, consider the sequences $(\mathit{foo};\mathit{bar}_1)$ and
$(\mathit{foo};\mathit{bar}_2)$ where
\begin{equation*}
\mathit{foo} = \left\{\begin{minipage}{.15\textwidth}
\begin{lstlisting}[mathescape=true,aboveskip=0mm,frame=none,numbers=none,
  xleftmargin=3mm]
$x \sim$ flip($1/3$)
\end{lstlisting}
\end{minipage} \right\}.
\end{equation*}
Let $\sigma_{ex}'$ be an output state where $x = \false, y = \true$, and let $\sigma_{ex}$ be
an arbitrary state.
The first sequence's transition semantics behave naturally for this output state:
\begin{align}
\dbracket{\mathit{foo};\mathit{bar}_1}_T(\sigma_{ex}' \mid \sigma_{ex}) = 2/3 \cdot 1/2 =
  1/3
  \label{eq:reject}
\end{align}

However, $(\mathit{foo};\mathit{bar}_2)$ represents a different distribution:
the {\tt observe} statement in $\mathit{bar}_2$ will disallow half of the
execution paths where $\mathit{foo}$ set $x = \true$. After the \texttt{observe}
statement is executed in $\mathit{bar}_2$, $\Pr(x = \true) = \frac{1}{2}$: the observation has increased the probability of $x$ being true in
$\mathit{foo}$, which was $1/3$. Thus, it is clear $\mathit{foo}$ and $\mathit{bar}_2$ cannot be
reasoned about solely as conditional probability distributions: \texttt{observe}
statements in $\mathit{bar}_2$ affect the conditional probability of
$\mathit{foo}$. Thus, the semantics of sequencing requires information beyond
solely the conditional probability of each of the sub-statements, as we discuss next.

\paragraph{Sequencing Semantics} The most interesting case in the semantics is sequencing.
We compute the transition semantics of sequencing $\dbracket{\ts_1; \ts_2}_T(\sigma' \mid
\sigma)$ using the rules of probability. To do this we require
the ability to compute the probability that a particular statement will not
violate an observation when beginning in state $\sigma$. Thus we introduce
a helper relation that we call the \emph{accepting semantics} (denoted
$\dbracket{\ts}_A$), which provides the probability that a given statement will
be {\em accepted} (i.e., that no \texttt{observe}s will fail) when executed from
a given initial state:
\begin{align*}
  \dbracket{\ts}_A : \Sigma \rightarrow [0,1]
\end{align*}

The accepting semantics is defined in Figure~\ref{fig:accepting}. The
first three rules in the figure are trivial. An {\tt observe} statement accepts
with probability 1 if the associated expression is true in the given state, and
otherwise with probability 0. A sequence of two statements accepts if both
statements accept, so the rule simply calculates that probability by summing
over all possible intermediate states. Last, the accepting probability of an
{\tt if} statement in state $\sigma$ is simply the accepting probability of
whichever branch will be taken from that state.

Now we can use the accepting semantics to give the transition semantics for
sequencing. First, we can compute the probability of both transitioning
from some initial state $\sigma$ to some final state $\sigma'$ \emph{and} the fact that no
observations are violated in $\ts_1$ or $\ts_2$:
\begin{align*}
  \alpha = \sum_{\tau \in \Sigma} \dbracket{\ts_1}_T(\tau \mid \sigma) \times
  \dbracket{\ts_2}_T(\sigma' \mid \tau) \times \dbracket{\ts_1}_A(\sigma) \times
  \dbracket{\ts_2}_A(\tau).
\end{align*}

In order to obtain the distribution of transitioning between states $\sigma$ and $\sigma'$ \emph{given} that no observations are violated, we must re-normalize this distribution by the probability that
no observations are violated:
\begin{align*}
  \beta = \dbracket{\ts_1}_A(\sigma) \times \sum_{\tau \in \Sigma} \dbracket{\ts_1}_T(\tau \mid \sigma) \times
  \dbracket{\ts_2}(\tau).
  \end{align*}
Thus, our conditional probability is $\dbracket{\ts_1; \ts_2}_T(\sigma' \mid \sigma) =
\frac{\alpha}{\beta}$. For completeness, we define the $0/0$ case to be equal to
$0$.

\section{Symbolic Compilation}

\begin{figure}
  \centering
  \begin{mdframed}
   \infrule[]{}{\texttt{skip} \rightsquigarrow (\gamma(V), \delta(V))}

   \infrule[]{\textrm{fresh} f}{x \sim \texttt{flip}(\theta) \rightsquigarrow\\
      \big((x' \Leftrightarrow f) \land \gamma(V \setminus
      \{x\}), \delta(V)[f \mapsto \theta, \neg f \mapsto 1-\theta]\big)}

    \infrule[]{}{x := \te \rightsquigarrow ((x' \Leftrightarrow \te)
      \land \gamma(V \setminus \{x\}), \delta(V))}

    \infrule[] {s_1 \rightsquigarrow (\varphi_1, w_1)  \andalso
      s_2 \rightsquigarrow (\varphi_2, w_2) \\
      \varphi_2' = \varphi_2[x_i \mapsto x_i', x_i' \mapsto x_i'']
    }{
    \ts_1; \ts_2 \rightsquigarrow ((\exists x_i' . \varphi_1 \land
    \varphi_2')[x_i '' \mapsto x_i'],  w_1 \uplus w_2)}

   \infrule[]{}{\texttt{observe(e)} \rightsquigarrow (\te \land \gamma(V), \delta(V))} 

   \infrule[]{s_1 \rightsquigarrow (\varphi_1, w_1) \andalso
     s_2 \rightsquigarrow (\varphi_2, w_2)
   }{
     \texttt{if}~\te~\{\ts_1\}~\texttt{else}~\{\ts_2\} \rightsquigarrow
     ((\te \land \varphi_1) \lor
     (\neg \te \land \varphi_2),  w_1 \uplus w_2)
   }

  \end{mdframed}
   \caption{Symbolic compilation rules.}
  \label{fig:symcomp}
\end{figure}

In this section we formally define our symbolic compilation of a \dippl{}
program to a weighted Boolean formula, denoted $\ts \rightsquigarrow (\varphi,
w)$. Intuitively, the formula $\varphi$ produced by the compilation represents
the program $\ts$ as a relation between initial states and final states, where
initial states are represented by unprimed Boolean variables $\{x_i\}$ and final
states are represented by primed Boolean variables $\{x_i'\}$. This is similar
to a standard encoding for model checking Boolean program, except we include
auxiliary variables in the encoding which are neither initial nor final
state variables \cite{Ball2000}.
These compiled weighted Boolean formula will have a probabilistic semantics
which allow them to be interpreted as either an accepting or transition
probability for the original statement.

Our goal is to ultimately give a correspondence between the compiled weighted
Boolean formula and the original denotational semantics of the statement.
First we define the translation of a state $\sigma$ to a logical formula:
\begin{definition}[Boolean state]
  Let $\sigma \!=\! \{(x_1,b_1),\ldots,(x_n,b_n)\}$.  We define the \emph{Boolean
state} $\form{\sigma}$ as $l_1 \wedge \ldots \wedge l_n$ where for each $i$,
$l_i$ is $x_i$ if $\sigma(x_i)=\true$ and $\neg x_i$ if $\sigma(x_i)=\false$.
For convenience, we also define a version that relabels state
variables to their primed versions, $\formp{\sigma} \triangleq \form{\sigma}[x_i
\mapsto x_i']$.
\end{definition}

Now, we formally describe how every compiled weighted Boolean formula can be
interpreted as a conditional probability by computing the appropriate weighted
model count:
\begin{definition}[Transition and accepting semantics]
  Let $(\varphi, w)$ be a weighted Boolean formula, and let $\sigma$ and $\sigma'$ be states.
  Then, the \emph{transition semantics} of
$(\varphi, w)$ is defined:
  \begin{align}
    \dbracket{(\varphi, w)}_T(\sigma' \mid \sigma) \triangleq
    \frac{\wmc(\varphi \land \form{\sigma} \land \formp{\sigma'}, w)}
    {\wmc(\varphi \land \form{\sigma}, w)} 
  \end{align}
  In addition the \emph{accepting semantics} of $(\varphi, w)$ is defined:
  \begin{align*}
    \dbracket{(\varphi, w)}_A(\sigma) \triangleq
    \wmc(\varphi \land \form{\sigma}, w).
  \end{align*}
\end{definition}
Moreover, the transition semantics of Definition~\ref{def:wmc-semantics} allow for more general queries to be phrased as WMC tasks as well. For example, the probability of some event $\alpha$ being true in the output state $\sigma'$ can be computed by replacing $\formp{\sigma'}$ in Equation~\ref{eq:wmc-semantics} by a Boolean formula for $\alpha$.


Finally, we state our correctness theorem, which describes the relation between
the accepting and transition semantics of the compiled WBF to the denotational
semantics of \dippl{}:
\begin{theorem}[Correctness of Compilation Procedure]
  Let $\ts$ be a \dippl{} program, $V$ be the set of all variables in $\ts$,
and  $\ts \rightsquigarrow (\varphi, w)$. Then for all states $\sigma$ and
$\sigma'$ over the variables in $V$, we have:
\begin{align}
    \dbracket{\ts}_T(\sigma' \mid \sigma) = \dbracket{(\varphi, w)}_T(\sigma' \mid \sigma)
      \end{align}
and
  \begin{align}
    \dbracket{\ts}_A(\sigma) = \dbracket{(\varphi, w)}_A(\sigma).
  \end{align}
\end{theorem}
\begin{proof}
  A complete proof can be found in Appendix~\ref{sec:app_correctness}.
\end{proof}

Theorem~\ref{thm:correctness} allows us to perform inference via weighted model
counting on the compiled WBF for a \dippl{} program. Next we give a description
of the symbolic compilation rules that satisfy this theorem.

\section{Symbolic Compilation Rules}
In this section we describe the symbolic compilation which satisfy
Theorem~\ref{thm:correctness} for each \dippl{} statement. The rules for
symbolic compilation are defined in Figure~\ref{fig:symcomp}. They rely on
several conventions. We denote by $V$ the set of all variables in the entire
program being compiled. If $V = \{x_1,\ldots,x_n\}$ then we use $\gamma(V)$ to
denote the formula $(x_1 \Leftrightarrow x_1') \wedge \ldots \wedge (x_n
\Leftrightarrow x_n')$, and we use $\delta(V)$ to denote the weight function
that maps each literal over $\{x_1,x_1',\ldots,x_n,x_n'\}$ to 1.

The WBF for {\tt skip} requires that the input and output states are equal and
provides a weight of 1 to each literal. The WBF for an assignment $x := \te$
requires that $x'$ be logically equivalent to $\te$ and all other variables'
values are unchanged. Note that \te is already a Boolean formula by the syntax
of \dippl{} so expressions simply compile to themselves. The WBF for drawing a
sample from a Bernoulli distribution, $x \sim \texttt{flip}(\theta)$, is similar
to that for an assignment, except that we introduce a (globally) fresh variable
$f$ to represent the sample and weight its true and false literals respectively
with the probability of drawing the corresponding value.

The WBE for an {\tt observe} statement requires the corresponding expression to
be true and that the state remains unchanged. The WBE for an {\tt if} statement
compiles the two branches to formulas and then uses the standard logical
semantics of conditionals. The weight function $w_1 \uplus w_2$ is a
\emph{shadowing union} of the two functions, favoring $w_2$. However, by
construction whenever two weight functions created by the rules have the same
literal in their domain, the corresponding weights are equal. Finally, the WBE
for a sequence composes the WBEs for the two sub-statements via a combination of
variable renaming and existential quantification.

In the following section, we delineate the advantages of utilizing WMC for
inference, and describe how WMC exploits program structure in order to perform
inference efficiently.

\section{Proof of Theorem~\ref{thm:correctness}}
\label{sec:app_correctness}
\subsection{Properties of WMC}
We begin with some important lemmas about weighted model counting:

\begin{lemma}[Independent Conjunction]
  \label{lem:wmcindconj}
  Let $\alpha$ and $\beta$ be Boolean sentences which share no variables. Then,
for any weight function $w$, $\wmc(\alpha \land \beta, w) = \wmc(\alpha, w)
\times \wmc(\beta, w)$.
\end{lemma}
\begin{proof}
  The proof relies on the fact that, if two sentences $\alpha$ and $\beta$ share
no variables, then any model $\omega$ of $\alpha \land \beta$ can be split into
two components, $\omega_\alpha$ and $\omega_\beta$, such that $\omega =
\omega_\alpha \land \omega_\beta$, $\omega_\alpha \Rightarrow \alpha$, and
$\omega_\beta \Rightarrow \beta$, and $\omega_\alpha$ and $\omega_\beta$ share
no variables. Then:
  \begin{align*}
    \wmc(\alpha \land \beta, w)
    =& \sum_{\omega \in \mods(\alpha \land \beta)} \prod_{l \in \omega} w(l)  \\
    =&\sum_{\omega_\alpha \in \mods(\alpha)} \sum_{\omega_\beta \in \mods(\beta)} \prod_{a \in \omega_\alpha} w(a)
       \times \prod_{b \in \omega_\beta } w(b)\\
    =&\left[  \sum_{\omega_\alpha \in \mods(\alpha)} \prod_{a \in \omega_\alpha} w(a) \right] \times \left[   \sum_{\omega_\beta \in \mods(\beta)} 
        \prod_{b \in \omega_\beta } w(b)\right]\\
    =& \wmc(\alpha, w) \times \wmc(\beta, w).
  \end{align*}
\end{proof}

\begin{lemma}
  \label{lem:wmcprod}
  Let $\alpha$ be a Boolean sentence and $x$ be a conjunction of literals.
  For any weight function $w$, $\wmc(\alpha,w) = \wmc(\alpha \mid x, w) \times
  \wmc(x,w)$. \footnote{The notation ``$\alpha \mid x$'' means condition $\alpha$
    on $x$.}
\end{lemma}
\begin{proof}
  Follows from Lemma~\ref{lem:wmcindconj} and the fact that $\alpha \mid x$ and
  $x$ share no variables by definition:
  \begin{align*}
    \wmc(\alpha \mid x) \times \wmc(x, w)
    =& \wmc((\alpha \mid x) \land x, w) & \text{By Lemma~\ref{lem:wmcindconj}} \\
    =& \wmc(\alpha, w).
  \end{align*}
\end{proof}

\begin{lemma}
  \label{lem:wmccond}
  Let $\alpha$ be a sentence, $x$ be a conjunction of literals, and $w$ be some
weight function. If for all $l \in x$ we have that $w(l) = 1$, then
$\wmc(\alpha \mid x, w) = \wmc(\alpha \land x, w)$. 
\end{lemma}
\begin{proof}
  \begin{align*}
    \wmc(\alpha \land x)
    =& \wmc((\alpha \land x) \mid x, w) \times \underbrace{\wmc(x, w)}_{=1} \\
    =& \wmc(\alpha \mid x, w).
  \end{align*}
\end{proof}

\begin{lemma}[Mutually Exclusive Disjunction]
  \label{lem:wmcmutex}
  Let $\alpha$ and $\beta$ be Boolean be mutually exclusive Boolean sentences
  (i.e., $\alpha \Leftrightarrow \neg \beta$). Then,
for any weight function $w$, $\wmc(\alpha \lor \beta, w) = \wmc(\alpha, w) +
\wmc(\beta, w)$.
\end{lemma}
\begin{proof}
    The proof relies on the fact that, if two sentences $\alpha$ and $\beta$ are
    mutually exclusive, then any model $\omega$ of $\alpha \lor \beta$ either
    entails $\alpha$ or entails $\beta$. We denote the set of models which
    entail $\alpha$ as $\Omega_\alpha$, and the set of models which entail
    $\beta$ as $\Omega_\beta$. Then,
    \begin{align*}
      \wmc(\alpha \lor \beta, w)
      =& \sum_{\omega_\alpha \in \Omega_\alpha} \sum_{\omega_\beta \in \Omega_\beta} \prod_{l \in \omega_\alpha} w(l)
      \prod_{l \in \omega_\beta}w(l)\\
      =& \wmc(\alpha, w) + \wmc(\beta, w).
    \end{align*}
\end{proof}

The following notion of functional dependency will be necessary for reasoning about
the compilation of the composition:

\begin{definition}[Functionally dependent WBF]
   Let $(\alpha, w)$ be a WBF, and let $\mathbf{X}$ and $\mathbf{Y}$ be two
variable sets which partition the variables in $\alpha$. Then we say that
$\mathbf{X}$ is \emph{functionally dependent on $\mathbf{Y}$ for $\alpha$} if for any total
assignment to variables in $\mathbf{Y}$, labeled $y$, there is at most one total assignment
to variables in $\mathbf{X}$, labeled $x$, such that $x \land y \models \alpha$.
\end{definition}

\begin{lemma}[Functionally Dependent Existential Quantification]
  \label{lem:wmcdet}
  Let $(\alpha, w)$ be a WBF with variable partition $\mathbf{X}$ and
$\mathbf{Y}$ such that $\mathbf{X}$ is functionally dependent on $\mathbf{Y}$
for $\alpha$.
Furthermore, assume that for any conjunction of literals $x$ formed from
$\mathbf{X}$, $\wmc(x) = 1$. Then,
$\wmc(\alpha) = \wmc(\exists \{x_i \in \mathbf{X}\} . \alpha)$.
\end{lemma}
\begin{proof}
  The proof follows from Lemma~\ref{lem:wmccond} and Lemma~\ref{lem:wmcmutex}.
  First, let $\mathbf{X} = x$ be a single variable, and assume all weighted
  model counts are performed with the weight function $w$. Then,
  \begin{align*}
    \wmc(\exists x .\alpha)
    =& \wmc((\alpha \mid x) \lor (\alpha \mid \neg x)) \\
    =& \wmc(\alpha \mid x) + \wmc(\alpha \mid \neg x) & \text{By mutual exclusion} \\
    =& \underbrace{\frac{1}{\wmc(x)}}_{=1}\wmc(\alpha \land x) +
       \underbrace{\frac{1}{\wmc(\neg x)}}_{=1} \wmc(\alpha \land \neg x) \\
    =& \wmc(\alpha \land x) + \wmc(\alpha \land \neg x) \\
    =& \wmc((\alpha \land x) \lor (\alpha \land \neg x)) \\
    =& \wmc(\alpha) & \text{By mutual exclusion}
  \end{align*}
  This technique easily generalizes to when $\mathbf{X}$ is a set of variables
  instead of a single variable.
\end{proof}

\subsection{Main Proof}

Let $\sigma$ and $\sigma'$ be an input and output state, let $V$ be the set of
variables in the entire program. The proof will proceed by induction on terms.
We prove the following inductive base cases for terms which are not defined
inductively.

\subsubsection{Base Cases}
  
\paragraph{Skip}
First, we show that the accepting semantics correspond. For any $\sigma$, we
have that $\dbracket{\texttt{skip}}_A(\sigma) = \wmc(\gamma(V) \land
\form{\sigma}) = 1$, since there is only a single satisfying assignment, which
has weight $1$. Now, we show that the transition semantics correspond:
\begin{itemize}[leftmargin=*]
\item Assume $\sigma' = \sigma$. Then,
  \begin{align*}
    \dbracket{\ts}_T(\sigma' \mid \sigma)
    =&\frac{\wmc(\gamma(V) \land \form{\sigma} \land
\formp{\sigma'}, \delta(V))}{\wmc(\gamma(V) \land
       \form{\sigma})}
    = 1
  \end{align*}
  since we have a single model in both numerator and denominator, both having weight $1$.
\item  Assume $\sigma \ne \sigma'$. Then:
    \begin{align*}
    \dbracket{\ts}_T(\sigma' \mid \sigma)
    =&\frac{\wmc(\gamma(V) \land \form{\sigma} \land
\formp{\sigma'}, \delta(V))}{\wmc(\gamma(V) \land
       \form{\sigma})}
    = 0
  \end{align*}
  since the numerator counts models of an unsatisfiable sentence.
\end{itemize} 

\paragraph{Sample}
Let $\varphi$ and $w$ be defined as in the symbolic compilation rules.
First we show that the accepting semantics correspond.
\begin{align*}
 \dbracket{x\sim\texttt{flip}(\theta)}_A(\sigma)
  =& \wmc((x' \Leftrightarrow f) \land \gamma(V \setminus \{x\}) \land \form{\sigma},
     w) \\
  =& \underbrace{\wmc(x' \Leftrightarrow f \mid \gamma(V \setminus \{x\}) \land \form{\sigma}, w)}
     _{=\theta+(1-\theta)=1}\times
     \underbrace{\wmc(\gamma(V \setminus \{x\}) \land \form{\sigma}, w)}_
   {=1, \text{ by def. of $\delta$}}& \text{By Lemma~\ref{lem:wmccond}} \\
  =& 1
\end{align*}

Now we show that the transition semantics correspond:
\begin{align*}
  \dbracket{{x\sim\texttt{flip}(\theta)}}_T(\sigma' \mid \sigma) 
  =&\wmc(\varphi \land \form{\sigma} \land \formp{\sigma'},w) \times \underbrace{\frac{1}{\wmc(\varphi \land \form{\sigma})}}_{=1}\\
  =&\wmc(\underbrace{\varphi \mid \form{\sigma} \land \formp{\sigma'}}_\alpha,w) \times 
    \underbrace{\wmc(\form{\sigma} \land \formp{\sigma'}}_{=1},w)
\end{align*}
We can observe the following about $\alpha$:
\begin{itemize}
\item If $\sigma = \sigma'[x \mapsto \true]$, then $\wmc(\varphi \land
  \form{\sigma} \land \formp{\sigma'},w) = \theta$.
\item If $\sigma = \sigma'[x \mapsto \false]$, then $\wmc(\varphi \land
  \form{\sigma} \land \formp{\sigma'},w) = 1-\theta$.
\item If $\sigma \ne \sigma'[x \mapsto \false]$ and
  $\sigma \ne \sigma'[x \mapsto \true]$ , $\alpha = \false$, so the weighted
  model count is 0.
\end{itemize}

\paragraph{Assignment} First we show that the accepting semantics correspond:
\begin{align*}
   \dbracket{x:=\te}_A(\sigma)
  =& \wmc((x' \Leftrightarrow \dbracket{\te}_\mathcal{S}) \land \gamma(V \setminus \{x\}) \land
     \form{\sigma}, w) \\
  =& 1,
\end{align*}
since there is exactly a single model and its weight is 1. Now, we show that the
transition semantics correspond:
\begin{align*}
  \dbracket{x:=\te}_T(\sigma' \mid \sigma)
  =& \wmc(\underbrace{(x' \Leftrightarrow \dbracket{\te}_\mathcal{S}) \land \gamma(V \setminus \{x\}) \land
     \form{\sigma} \land \form{\sigma'}}_\alpha, w)
     \times \underbrace{\frac{1}{\wmc((x' \Leftrightarrow
     \dbracket{\te}_\mathcal{S}) \land \gamma(V \setminus \{x\}) \land
     \form{\sigma})}}_{=1}
\end{align*}

\begin{itemize}
\item Assume $\sigma = \sigma'[x \mapsto \dbracket{\te}(\sigma)]$. Then,
  $\alpha$ has a single model, and the weight of that model is 1, so
  $\wmc(\alpha, w) = 1$.
\item Assume $\sigma \ne \sigma'[x \mapsto \dbracket{\te}(\sigma)]$. Then,
  $\alpha$ is unsatisfiable, so $\wmc(\alpha, w) = 0$.
\end{itemize}

\paragraph{Observe} First we prove the transition semantics correspond:
\begin{align*}
  \dbracket{\texttt{observe}(\te)}_A(\sigma)
  =& \wmc(\dbracket{\te}_\sym \land \gamma(V) \land \form{\sigma}, w) \\
  =&
     \begin{cases}
       1 \quad& \text{if } \form{\sigma} \models \dbracket{\ts}_\sym \\
       0 & \text{otherwise}
     \end{cases}
\end{align*}

Now, we can prove that the transition semantics correspond:

\begin{align*}
  \dbracket{\texttt{observe}(\te)}_T(\sigma' \mid \sigma)
  =& \wmc(\underbrace{\dbracket{\te}_\sym \land \gamma(V) \land \form{\sigma} \land \formp{\sigma'}}_\alpha, w) \times
     \frac{1}{\wmc(\underbrace{\dbracket{\te}_\sym \land \gamma(V) \land \form{\sigma}}_\beta, w)}
\end{align*}

We treat a fraction $\frac{0}{0}$ as $0$. Then, we can apply case analysis:
\begin{itemize}
\item Assume $\sigma = \sigma'$ and $\dbracket{\te}(\sigma) = \true$. Then,
  both $\alpha$ and $\beta$ have a single model with weight 1, so
  $\dbracket{\texttt{observe}(\te)}_T(\sigma' \mid \sigma) = 1$.
\item Assume $\sigma \ne \sigma'$ or $\dbracket{\te}(\sigma) \ne \true$. Then,
  either $\form{\sigma} \land \dbracket{\te}_\sym \models \false$ or $\gamma(V)
  \land \form{\sigma} \land \formp{\sigma'} \models \false$; in either case,
  $\dbracket{\texttt{observe}(\te)}_T(\sigma' \mid \sigma) = 0$.
\end{itemize}

\subsubsection{Inductive Step}
Now, we utilize the inductive hypothesis to prove the theorem for the
inductively-defined terms. Formally, let $\ts$ be a \dippl{} term, let $\{\ts_i\}$
be sub-terms of $\ts$.
Then, our inductive hypothesis states that for each sub-term $\ts_i$ of $\ts$,
where $\ts_i \rightsquigarrow (\varphi, w)$, we have that for any two states
$\sigma, \sigma'$, $\dbracket{\ts_i}_T(\sigma' \mid \sigma) = \dbracket{(\varphi
  ,w)}_T(\sigma') \mid \sigma)$ and
$\dbracket{\ts_i}_A(\sigma) = \dbracket{(\varphi, w)}_A(\form{\sigma})$. Then,
we must show that the theorem holds for $\ts$ using this hypothesis.

\begin{remark}
  For the inductively defined compilation semantics, the weight function $w =
w_1 \uplus w_2$ is a unique and well-defined weight function, since the only
source of weighted variables is from a \texttt{flip} term, which only assigns a
weight to fresh variables; thus, there can never be a disagreement between the
two weight functions $w_1$ and $w_2$ about the weight of a particular variable.
\end{remark}

\paragraph{Composition} Let $\varphi, w, \varphi_1, \varphi_2, \varphi_2', w_1,$ and
$w_2$ be defined as in the symbolic compilation rules. By the inductive
hypothesis, we have that the theorem holds for $(\varphi_1, w_1)$ and
$(\varphi_2, w_2)$. We observe that the weighted model counts of $\varphi_2$ are
invariant under relabelings. I.e., for any states $\sigma, \sigma', \sigma''$:
\begin{align*}
  \wmc(\varphi_2 \land \form{\sigma}) =& \wmc(\varphi_2' \land \formp{\sigma}) \\
  \wmc(\varphi_2 \land \form{\sigma} \land \form{\sigma'})
  =& \wmc(\varphi_2' \land \formp{\sigma} \land \formpp{\sigma'})
\end{align*}
where $\formpp{\cdot}$ generates double-primed state variables. Now we show that
the WBF compilation has the correct accepting semantics, where each weighted
model count implicitly utilizes the weight function $w$:

\begin{align*}
  \dbracket{\ts_1; \ts_2}_A(\sigma) =&
    \dbracket{\ts_1}_A(\sigma) \times
    \sum_{\tau \in \Sigma} \left(  \dbracket{\ts_1}_T  (\tau \mid \sigma) \times \dbracket{\ts_2}_A(\tau) \right) \\
   =& Z \times \sum_\tau \frac{\wmc(\varphi_1 \land \form{\sigma} \land \formp{\tau},
     w)}{Z} \times
      \wmc(\varphi_2' \land \formp{\tau}) \quad\text{ where } Z = \wmc(\varphi_1 \land \form{\sigma}) \\
  =& \sum_\tau \wmc(\varphi_1 \land \form{\sigma} \land \formp{\tau}) \times
     \wmc(\varphi_2' \land \formp{\tau}) \\
  =& \sum_\tau \wmc(\varphi_1 \land \form{\sigma} \mid \formp{\tau}) \times 
     \wmc(\varphi_2' \mid \formp{\tau}) \times
     \underbrace{[\wmc(\formp{\tau})]^2}_{=1}
    \tag{By Lemma~\ref{lem:wmcprod}}\\
  =& \sum_\tau\wmc\left(\big[\varphi_1 \land \form{\sigma} \mid \formp{\tau} \big] \land
     \big[\varphi_2'  \mid \formp{\tau} \big]\right)
    \tag{By Lemma~\ref{lem:wmcindconj}}\\
  =& \sum_\tau\wmc\left(\big[\varphi_1 \land \form{\sigma} \land
     \varphi_2'  \mid \formp{\tau}\big] \right)
    \\
  =& \sum_\tau\wmc\left(\varphi_1 \land \form{\sigma} \land
     \varphi_2' \land \formp{\tau}\right) \times \underbrace{\frac{1}{\wmc(\formp{\tau})}}_{=1}
    \tag{By Lemma~\ref{lem:wmcprod}} \\
  \label{eq:step3}
  =& \sum_\tau \wmc(\varphi_1 \land \varphi_2' \land\form{\sigma} \land \formp{\tau}
     )\\
  =& \wmc\left(\bigvee_\tau \varphi_1 \land \varphi_2' \land\form{\sigma} \land \formp{\tau}
     \right) \tag{By Lemma~\ref{lem:wmcmutex}} \\
  =& \wmc\left( \varphi_1 \land \varphi_2' \land\form{\sigma} \land \left[ \bigvee_\tau\formp{\tau} \right]
     \right)   \\
  =& \wmc\left( \varphi_1 \land \varphi_2' \land\form{\sigma}
     \right) \\
  =& \wmc\left(\exists \{x_i'\} .  \varphi_1 \land \varphi_2' \land\form{\sigma}
     \right) \tag{By Lemma~\ref{lem:wmcdet}} \\
  =& \wmc((\exists \{x_i'\} . \varphi_2' \land \form{\sigma}[x'' \mapsto x']).
\end{align*}

Now, we can prove the transition semantics correspond for composition, where all
model counts are implicitly utilizing the weight function $w$:

\begin{align*}
  \dbracket{\ts_1; \ts_2}_T(\sigma' \mid \sigma)
  &= \frac{\sum_{\tau \in \Sigma} \dbracket{\ts_1}_T(\tau \mid \sigma) \times
    \dbracket{\ts_2}_T(\sigma' \mid \tau) \times \dbracket{\ts_2}_A(\tau)}
    {\sum_{\tau \in \Sigma} \dbracket{\ts_1}_T(\tau \mid \sigma) \times \dbracket{\ts_2}_A(\tau) }\\
  &= \frac{\sum_{\tau \in \Sigma} \dbracket{\ts_1}_T(\tau \mid \sigma) \times
    \dbracket{\ts_2}_T(\sigma' \mid \tau) \times \dbracket{\ts_2}_A(\tau)}
    {\frac{1}{\dbracket{\ts_1}_A(\sigma)} \times \underbrace{\dbracket{\ts_1}_A(\sigma) \times
    \sum_{\tau \in \Sigma} \dbracket{\ts_1}_T(\tau \mid \sigma) \times \dbracket{\ts_2}_A(\tau)}_{=\dbracket{\ts_1;\ts_2}_A(\sigma)} }
  \\
  &= \frac{\sum_{\tau \in \Sigma} \frac{\wmc(\varphi_1 \land \form{\sigma} \land \formp{\tau})}{\wmc(\varphi_1 \land \form{\sigma})}
    \times
    \wmc(\varphi_2' \land \formpp{\sigma'} \land \formp{\tau})}
    {\frac{1}{\wmc(\varphi_1 \land \form{\sigma})}\wmc(\varphi \land \form{\sigma})}
  \tag{By inductive hyp.}
  \\
  &= \frac{1}{\wmc(\varphi \land \form{\sigma})} \times
    \sum_{\tau \in \Sigma} \wmc(\varphi_1 \land \form{\sigma} \land \formp{\tau}) \times 
    \wmc(\varphi_2' \land \formpp{\sigma'} \land \formp{\tau})
  \\
  &= \frac{1}{\wmc(\varphi \land \form{\sigma})} \times
    \sum_{\tau \in \Sigma} \wmc(\varphi_1 \land \form{\sigma} \mid \formp{\tau}) \times 
    \wmc(\varphi_2' \land \formpp{\sigma'} \mid \formp{\tau}) \times
    \underbrace{[\wmc(\formp{\tau})]^2}_{=1}
    \tag{By Lemma~\ref{lem:wmccond}}
  \\
  &= \frac{1}{\wmc(\varphi \land \form{\sigma})} \times
    \sum_\tau \wmc\left(\varphi_1 \land \varphi_2' \land \form{\sigma} \land \formpp{\sigma'} \mid \formp{\tau}\right)
  \tag{By Lemma~\ref{lem:wmcprod}}
  \\
  &= \frac{1}{\wmc(\varphi \land \form{\sigma})} \times
    \sum_\tau \wmc\left(\varphi_1 \land \varphi_2' \land \form{\sigma} \land \formpp{\sigma'} \land \formp{\tau}\right)
    \times \frac{1}{\wmc(\formp{\tau})}
  \tag{By Lemma~\ref{lem:wmccond}}
  \\
  &= \frac{1}{\wmc(\varphi \land \form{\sigma})} \times
    \wmc\left(\bigvee_\tau \varphi_1 \land \varphi_2' \land \form{\sigma} \land \formpp{\sigma'} \land \formp{\tau}  \right)
  \tag{By Lemma~\ref{lem:wmcmutex}}
  \\
  &= \frac{1}{\wmc(\varphi \land \form{\sigma})} \times
    \wmc\left(\varphi_1 \land \varphi_2' \land \form{\sigma} \land \formpp{\sigma'} \land \left[\bigvee_\tau \formp{\tau} \right] \right)
  \\
  &= \frac{1}{\wmc(\varphi \land \form{\sigma})} \times
    \wmc\left(\varphi_1 \land \varphi_2' \land \form{\sigma} \land \formpp{\sigma'} \right)
  \\
  &= \frac{\wmc\left(\exists \{x_i'\}. \varphi_1 \land \varphi_2' \land \form{\sigma} \land \formpp{\sigma'} \right)}{\wmc(\exists \{x_i'\}. \varphi_1 \land \varphi_2' \land \form{\sigma})} 
  &
  \tag{By Lemma~\ref{lem:wmcdet}}
\end{align*}

\paragraph{\texttt{if}-statements}
Let $\varphi_1, \varphi_2, w, \varphi$ be defined as in the compilation rules.
First, we prove that the accepting semantics correspond:
\begin{align*}
  \dbracket{\texttt{if}(\te)~\{\ts_1\}~\texttt{else}~\{\ts_2\}}_A(\sigma)
  =& \begin{cases}
    \dbracket{\ts_1}_A(\sigma) ~~& \text{if } \dbracket{\te}(\sigma) = \true \\
    \dbracket{\ts_2}_A(\sigma) ~~& \text{if } \dbracket{\te}(\sigma) = \false \\
  \end{cases} \\
  =& \begin{cases}
    \wmc(\varphi_1 \land \form{\sigma})~~& \text{if } \dbracket{\te}_\sym \land \form{\sigma} \models \true\\
    \wmc(\varphi_2 \land \form{\sigma})~~& \text{otherwise}
     \end{cases} & \text{By Inductive Hyp.}\\
  =&                         
     \wmc((\dbracket{\te}_\sym \land \varphi_1 \land \form{\sigma}) \lor (\neg \dbracket{\te}_\sym \land \varphi_2 \land \form{\sigma}))
                                 & (\dagger) \\
  =& \wmc\left(\big[(\dbracket{\te}_\sym \land \varphi_1) \lor (\neg \dbracket{\te}_\sym \land \varphi_2)\big] \land \form{\sigma}\right)
\end{align*}
Where $(\dagger)$ follows from Lemma~\ref{lem:wmcmutex} and the mutual
exclusivity of $\dbracket{\te}_\sym$ and $\neg\dbracket{\te}_\sym$. Now we can
prove the transition semantics correspond:
\begin{align*}
    \dbracket{\texttt{if}(\te)~\{\ts_1\}~\texttt{else}~\{\ts_2\}}_T(\sigma' \mid \sigma)
  =& \begin{cases}
    \dbracket{\ts_1}_T(\sigma' \mid \sigma) ~~& \text{if } \dbracket{\te}(\sigma) = \true \\
    \dbracket{\ts_2}_T(\sigma' \mid \sigma) ~~& \text{if } \dbracket{\te}(\sigma) = \false \\
  \end{cases} \\
  =& \begin{cases}
    \frac{\wmc(\varphi_1 \land \form{\sigma} \land \formp{\sigma'})}{\wmc(\varphi_1 \land
      \form{\sigma})}~~& \text{if } \dbracket{\te}_\sym \land \form{\sigma} \models
    \true\\
    \frac{\wmc(\varphi_2 \land \form{\sigma} \land \formp{\sigma'})}{\wmc(\varphi_2 \land
      \form{\sigma})}
    ~~& \text{otherwise}
  \end{cases} & \text{By Inductive Hyp.}\\
  =& \frac{\wmc\left(\big[\dbracket{\te}_\sym \land \varphi_1 \land \form{\sigma} \land \formp{\sigma'}\big] \lor
     \big[\neg\dbracket{\te}_\sym \land \varphi_2 \land \form{\sigma} \land \formp{\sigma'} \big]\right)}
     {\wmc\left(\big[(\dbracket{\te}_\sym \land \varphi_1) \lor (\neg
     \dbracket{\te}_\sym \land \varphi_2)\big] \land \form{\sigma}\right)}\\
  =& \frac{\wmc\left(\big[(\dbracket{\te}_\sym \land \varphi_1) \lor
     (\neg\dbracket{\te}_\sym \land \varphi_2) \big] \land \form{\sigma} \land \formp{\sigma'} \right)}
     {\wmc\left(\big[(\dbracket{\te}_\sym \land \varphi_1) \lor (\neg
     \dbracket{\te}_\sym \land \varphi_2)\big] \land \form{\sigma}\right)}
\end{align*}

This concludes the proof.


\end{document}